%% file: fo-poly.tex
\documentclass[conference]{IEEEtran}
\usepackage[utf8]{inputenc}

\usepackage{booktabs}
\usepackage{amssymb}
\usepackage{stmaryrd}
\usepackage{amsfonts}
\usepackage{mathtools}
\usepackage{thmtools}
\usepackage{amsthm}
\usepackage{dblfloatfix}
\usepackage[ruled,linesnumbered]{algorithm2e}
\usepackage{tikz}
\usetikzlibrary{automata,positioning,arrows,backgrounds}
\usepackage{microtype}
\usepackage{hyperref}
\usepackage[capitalize,noabbrev,nameinlink]{cleveref}
\usepackage{multirow}

\newcommand{\intro}[1]{#1}
\newcommand{\kl}[1]{#1}
\input{globals/macros}

\title{$\Rel$-polyregular functions}

\author{\IEEEauthorblockN{Thomas Colcombet\IEEEauthorrefmark{1},
Gaëtan Douéneau-Tabot\IEEEauthorrefmark{1}\IEEEauthorrefmark{2} and
Aliaume Lopez\IEEEauthorrefmark{1}\IEEEauthorrefmark{3}}
\IEEEauthorblockA{\IEEEauthorrefmark{1}Université Paris Cité, CNRS, IRIF, F-75013, Paris, France}
\IEEEauthorblockA{\IEEEauthorrefmark{2}
Direction générale de l’armement - Ingénierie des projets, Paris, France}
\IEEEauthorblockA{\IEEEauthorrefmark{3}ENS Paris-Saclay, CNRS, LMF, France}}

    \def\subparagraph{\subsubsection}
    \usepackage{thm-restate}

\begin{document}

\maketitle

\input{parts/abstract.tex}


\input{parts/intro}

\input{parts/prelim}

\input{parts/Zrat.tex}

\input{parts/esqueletto}

\input{parts/residual}


\input{parts/aperiodic}

\input{parts/conclusion}

\newpage

\bibliography{globals/ressources}




\clearpage
\appendices\pagenumbering{roman}
\setlength\textwidth{390pt}
\setlength\oddsidemargin{31pt}
\onecolumn

\input{parts/appendix.tex}

\end{document}

%% file: globals/macros.tex
\theoremstyle{plain}
\newtheorem{theorem}{Theorem}[section]
\newtheorem{lemma}[theorem]{Lemma}
\newtheorem{proposition}[theorem]{Proposition}
\newtheorem{corollary}[theorem]{Corollary}

\newtheorem{example}[theorem]{Example}
\newtheorem{definition}[theorem]{Definition}
\newtheorem{remark}[theorem]{Remark}

\newtheorem{claim}[theorem]{Claim}
\newtheorem{inductionstep}[theorem]{Induction Step}
\newenvironment{claimproof}{\begin{proof}}{\end{proof}}

\usepackage{subcaption}

\usepackage{varwidth}

\newcommand{\citep}[1]{\cite{#1}}

\newcommand{\defined}{\coloneqq}

\newcommand{\set}[1]{\{ {#1} \}}

\newcommand{\setof}[2]{\{ {#1}~\colon {#2} \}}

\renewcommand{\implies}{\Rightarrow}

\newcommand{\Nat}{\mathbb{N}}
\newcommand{\Rel}{\mathbb{Z}}
\newcommand{\Rea}{\mathbb{R}}
\newcommand{\Com}{\mathbb{C}}
\newcommand{\Uni}{\mathbb{U}}
\newcommand{\Rat}{\mathbb{Q}}
\newcommand\id{\operatorname{id}}
\renewcommand{\phi}{\varphi}

\newcommand\Mat[1]{\mathop{{\mathcal{M}^{#1}}}}

\newcommand{\polysum}{\operatorname{\mathsf{sum}}}
\newcommand{\Dom}{\operatorname{\mathsf{Dom}}}

\newcommand\card[1]{{\#}{#1}}
\newcommand\MSO{{\operatorname{\mathsf{MSO}}}}
\newcommand\INV{{\operatorname{\mathsf{INV}}}}
\newcommand\FO{{\operatorname{\mathsf{FO}}}}

\newcommand\Ve{\mathsf{\Rel Poly}}
\newcommand\We{\mathsf{\Rel SF}}
\newcommand\VeN{\mathsf{\Nat Poly}}
\newcommand\WeN{\mathsf{\Nat SF}}

\newcommand\Vect[1]{\operatorname{{\mathsf{Span}}}_{#1}}
\newcommand\Spec{\operatorname{{\mathsf{Spec}}}}
\newcommand\existsexactly[1]{\mathop{{\exists^{={#1}}}}}
\newcommand\forests[2]{\operatorname{{\mathcal{F}}}_{#2}^{#1}}
\newcommand\repword{\operatorname{{\mathsf{word}}}}
\newcommand\repforest{\operatorname{{\mathsf{forest}}}}
\newcommand\Ske{\operatorname{{\mathsf{Skel}}}}
\newcommand\findep{{f_\mathsf{indep}}}
\newcommand\fdep{{f_\mathsf{dep}}}
\newcommand\push[2]{{#2}\mathop{{\triangleright}}{#1}}

\newcommand\Res{\operatorname{{\mathsf{Res}}}}
\newcommand\polysim[1]{\mathrel{{\sim_{#1}}}}
\newcommand\partitionskel{\operatorname{{\mathsf{skel-root}}}}
\newcommand{\msodep}{\operatorname{{\mathsf{{sym-dep}}}}}
\newcommand{\msoisleaf}{\operatorname{{\mathsf{isleaf}}}}
\newcommand{\msoprod}{\operatorname{{\mathsf{between}}}}
\newcommand{\msoleft}{\operatorname{{\mathsf{left}}}}
\newcommand{\msoright}{\operatorname{{\mathsf{right}}}}
\newcommand{\skeltype}{\operatorname{{\mathsf{s}-\mathsf{type}}}}
\newcommand{\skelstart}{\operatorname{{\mathsf{start}}}}
\newcommand{\skelend}{\operatorname{{\mathsf{end}}}}
\newcommand{\Types}{\operatorname{{\mathsf{Types}}}}
\newcommand{\skelhastype}[1]{\operatorname{{\mathsf{has}-\mathsf{s}-\mathsf{type}}}_{#1}}
\newcommand\dependson{\mathop{\operatorname{{\mathsf{depends}-\mathsf{on}}}}}

\newcommand{\fonc}{\rightarrow}
\newcommand{\bigO}{\mathcal{O}}
\newcommand{\Oras}{\mathcal{H}}

\newcommand{\Aper}{\mathcal{A}}

\newcommand{\trans}{\mathcal{T}}
\newcommand{\movi}{\varepsilon}
\newcommand{\vide}{\varnothing}

\newcommand{\tree}[1]{\langle #1 \rangle}
\newcommand{\car}[1]{\mathbf{1}_{#1}}

\newcommand\cauchy{\mathop{{\otimes}}}

\newcommand{\nod}{\mathfrak{t}}
\newcommand{\nodb}{\mathfrak{s}}
\newcommand{\Nodes}{\operatorname{\mathsf{Nodes}}}
\newcommand{\Leaves}{\operatorname{\mathsf{Leaves}}}

\newcommand{\lexl}{<_{\text{lex}}}

%% file: parts/abstract.tex
\begin{abstract}
	This paper studies a robust class of functions
	from finite words to integers that we call
	$\Rel$-polyregular functions. We show that
	it admits natural characterizations in terms of logics,
	$\Rel$-rational expressions,
    $\Rel$-rational series and transducers.

	We then study two subclass membership problems.
	First, we show that the asymptotic growth rate
	of a function is computable, and corresponds
	to the minimal number of variables required to represent it
	using logical formulas. Second, we show that first-order
	definability of $\Rel$-polyregular functions is decidable. To show the latter, we introduce an original notion
	of residual transducer, and provide a semantic characterization
	based on aperiodicity.
\end{abstract}

%% file: parts/intro.tex
\section{Introduction}
\label{sec:intro}

\begin{table*}[!h]
    \centering
    \begin{tabular}{p{3.5cm}p{5.7cm}p{5.7cm}}
        \toprule
        \textbf{Formalism} & \textbf{Characterization of $\Ve$}
                             & \textbf{Characterization of $\We$} \\
        \midrule
        Counting formulas 
        &
        Counting valuations in $\MSO$
        (\cref{def:Z-poly})
                          &
        Counting valuations in $\FO$
        (\cref{def:Zsf})
        \\
        \midrule
        \addlinespace
        Polyregular functions 
                          &
                          $\polysum \circ$ polyregular
                          (\cref{prop:mikolaj1})
                          &
                           $\polysum \circ$ star-free polyregular
                          (\cref{prop:mikolaj2})
      \\
      \addlinespace
      \midrule
        $\Rel$-rational expressions 
        &
        Closure of rational languages under
         Cauchy products, sums, and $\Rel$-products
        (\cref{theo:rational-polygrowth})
                                    &
        Closure of star-free languages under
         Cauchy products, sums, and $\Rel$-products
        (\cref{theo:Zsf-expressions})
        \\
      \addlinespace
      \midrule
        &
        Ultimately $N$-polynomial (\cref{theo:Zpoly-Zrat})
                                                         &
        Ultimately 
        $1$-polynomial (\cref{theo:big-induction})
        \\
        \addlinespace
        {$\Rel$-rational series that are/have}  &
        Polynomial growth (\cref{theo:Zpoly-Zrat})
                                                         &
        n/a
        \\
        \addlinespace
                                                         &
        Eigenvalues in $ \set{0} \cup \Uni$ (\cref{theo:Zpoly-Zrat}) 
                                                         & 
        Eigenvalues in $\set{0, 1}$ (\cref{theo:Zsf-eigen})
        \\
        \addlinespace
        \midrule
        Residual transducer 
        &
            Residual transducer (\cref{cor:Vk-trans})
                          &
            Counter-free residual transducer
            (\cref{theo:big-induction})
        \\
        \bottomrule
    \end{tabular}
    \caption{\label{tab:summary-characterisations}
        Summary of the characterizations of $\Ve$ and $\We$
        expressed in different formalisms.}
\end{table*}

Deterministic finite state automata define
the well-known and robust class of regular languages.
This class is captured by different formalisms
such as expressions (regular expressions \cite{kleene1956representation}),
logic (Monadic Second Order ($\MSO$) logic
 \cite{buchi1960weak}), and algebra (finite monoids  \cite{schutzenberger1961definition}).
It contains a robust
subclass of independent interest: star-free regular languages,
which admits equivalent descriptions in terms of machines
(counter-free automata \cite{mcnaughton1971counter}),
expressions (star-free expressions \cite{schutzenberger1965finite}),
logic (first-order ($\FO$) logic \cite{perrin1986first})
and algebra (aperiodic monoids \cite{schutzenberger1965finite}).
Furthermore,
one can decide if a regular language is
star-free, and the proof relies on the existence
(and computability) of a canonical object
associated to each language (its minimal automaton \cite{mcnaughton1971counter}
or, equivalently, its syntactic monoid \cite{schutzenberger1965finite}).

Numerous works have attempted to
carry the notion of regularity from languages
to word-to-word functions.
This work led to a plethora of non-equivalent classes
(such as sequential, rational, regular
and polyregular functions \cite{bojanczyk2018polyregular}).
Decision problems, including first-order definability, become
more difficult and more interesting for functions \cite{scott1967some},
mainly due to the lack of canonical objects similar to the minimal automata
of regular languages. It was shown recently
that first-order definability is decidable for the class
of rational functions~\cite{filiot2016aperiodicity}
and that a canonical object can be built~\cite{filiot2018canonical}.

%

This paper is a brochure for a natural
class of functions from finite words
to integers, which we name \emph{$\Rel$-polyregular functions}.
Its definition stems from the logical description of regular languages.
Given an $\MSO$ formula $\varphi(\vec{x})$
with free first-order variables $\vec{x}$, 
and a word $w \in A^*$,
we define $\card \varphi(w)$ to be the number
of valuations $\nu$ such that $w, \nu \models \varphi(\vec{x})$.
The indicator functions of regular languages are exactly
the functions $\card \phi$ where $\phi$ is a sentence
(i.e. it does not have free variables, hence has at most one valuation:
the empty one).
We define the class of $\Rel$-polyregular functions, denoted $\Ve$,
as the class of $\Rel$-linear combinations of  functions $\card{\varphi}$
where $\phi$ is in $\MSO$ with first-order free variables.

The goal of this paper is to advocate for the robustness 
of $\Ve$. To that end, we shall provide numerous characterizations of
these functions and relate them to pre-existing
models. We also solve several membership problems
and provide effective conversion algorithms.
This equips $\Ve$ with a smooth and elegant
theory, which subsumes that of regular languages.
%

\subparagraph*{Contributions}
We introduce the class $\Ve$ 
as a natural generalization of regular
languages via simple counting of $\MSO$ valuations.
This definition can be seen as 
a restricted version of the Quantitative $\MSO$ introduced in
\cite{kreutzer2013}.
It also coincides with the linear finite counting automata of
\cite{schutzenberger1962}.
We first connect $\Rel$-polyregular functions
to word-to-word polyregular functions
\cite{bojanczyk2018polyregular},
providing a justification for their name.
As a class of functions from finite words to integers,
it is then natural to compare $\Ve$ with the well-studied
class of $\Rel$-rational series (see e.g. \cite{berstel2011noncommutative}).
We observe that $\Ve$ is exactly the subclass of
$\Rel$-rational series that have polynomial growth,
i.e. the functions such that $|f(w)| = \bigO(|w|^k)$
for some $k \ge 0$, 
by making effective the results
of Schützenberger \cite{schutzenberger1962}. As a consequence, we provide
a simple syntax of $\Rel$-rational expressions
to describe $\Ve$ as those built without the Kleene star.
We also show how $\Ve$
can be described using natural restrictions
on the eigenvalues of representations of $\Rel$-rational series.
This property is built upon a quantitative pumping lemma
characterizing the ultimate behavior of $\Rel$-polyregular
functions as ``ultimately $N$-polynomial'' for some $N \geq 0$.
We summarize these results
in the second column of \cref{tab:summary-characterisations}.

We then refine the description of $\Ve$ by
considering for all $k \ge 0$,
the class $\Ve_k$ of functions
described using at most $k$ free
variables in the counting $\MSO$ formulas.
It is easy to check that if $f \in \Ve_k$
then $|f(w)| = \bigO(|w|^k)$. Our first
main theorem shows that this property
is a sufficient and necessary condition
for a function of $\Ve$ to be in $\Ve_k$
 (see \cref{fig:Z-classes}). This result is an analogue
of the various ``pebble minimization theorems''
that were shown for word-to-word polyregular functions
\cite{doueneau2021pebble,bojanczyk2022transducers,
doueneau2023pebble,bojanczyk2022growth}.
We also prove that the membership 
problem of $\Ve_k$ inside $\Ve$ is decidable.

Our second main contribution is the definition
of an almost canonical object associated to each
function of $\Ve$. We name this object the 
\emph{residual transducer} of the function,
and show that it can effectively be built.
Its construction is inspired by the
residual automaton of a regular language,
and heavily relies on the decision procedure from
$\Ve$ to $\Ve_k$.

Finally, we define the class $\We$ of
\emph{star-free $\Rel$-polyregular functions},
as the class of linear combinations of $\card \varphi$
where $\varphi$ is a first-order formula with free
first-order variables. As in the case
of $\Ve$, observe that the indicator functions of 
star-free languages are exactly the $\card \varphi$
where $\varphi$ is a first-order sentence.
Our third main contribution then applies the construction
of the residual transducer to show that the membership
problem from $\Ve$ to $\We$ is decidable.
Incidentally, we introduce for $k \ge 0$ the class $\We_k$
(defined in a similar way as $\Ve_k$) and show that $\We_k = \We \cap \Ve_k$,
as depicted in \cref{fig:Z-classes}.
Furthermore, we show that the numerous characterizations
of $\Ve$ in terms of existing models can naturally be
specialized to build characterizations of $\We$, as depicted in the third
column of \cref{tab:summary-characterisations}.

Overall, our contribution is the 
description of a natural theory
of functions from finite words to $\Rel$, that is the consequence of
a reasonable computational power (polynomial growth,
i.e. less than $\Rel$-rational series) and the ability
to correct errors during a computation
(using negative numbers). Furthermore, the theory
of $\Rel$-polyregular functions is built using new
and non-trivial proof techniques.

\subparagraph*{Outline}
\Cref{sec:prelim} is devoted to the introduction of the classes $\Ve$
and $\Ve_k$. We also compare $\Ve$ with polyregular functions
and with $\Rel$-rational series.
We then devote \cref{sec:pebblemin} to a free variable minimization
theorem (\cref{thm:skel:pebblemin}), which is a key
result towards the effective computation of a canonical 
residual transducer in \cref{sec:residual}.
We then introduce $\We$ and $\We_k$
in \cref{sec:aperiodic}, and use the residual transducer
to prove the decidability of $\We$ inside $\Ve$
(\cref{theo:FO-decidable}).
We conclude by connecting $\We$ to polyregular functions
and $\Rel$-rational series.
All of the aforementioned results include algorithms to decide membership
and provide effective conversions between the various representations.

\input{parts/fig-classes.tex}

%% file: parts/fig-classes.tex
\begin{figure*}[!t]

	\begin{center}
		\hspace*{-0.4cm}
        \resizebox{0.9\linewidth}{!}{
		\begin{tikzpicture}
			\def\gellipse{(0,2.4) ellipse (4cm and 3cm)}
			\def\aellipse{(0,1.8) ellipse (3.5cm and 2.5cm)}
			\def\bellipse{(0,1.2) ellipse (3cm and 1.9cm)}
			\def\cellipse{(0,0.6) ellipse (2.5cm and 1.3cm)}
			\def\dellipse{(0,0) ellipse (1.7cm and 0.7cm)}
			\fill[fill = blue!7] \gellipse;
			\fill[fill = blue!14]  \aellipse;
			\fill[fill = blue!26] \bellipse ;
			\fill[fill = blue!38]  \cellipse;
			\fill[fill = blue!50] \dellipse;

			\def\rellipse{(-2.2,1.8) ellipse (2.1cm and 3cm)}
			\begin{scope}
				\clip {\aellipse};
				\fill[red!10] \rellipse;
			\end{scope}

			\begin{scope}
				\clip {\bellipse};
				\fill[red!20] \rellipse;
			\end{scope}

			\begin{scope}
				\clip {\cellipse};
				\fill[red!30] \rellipse;
			\end{scope}

			\begin{scope}
				\clip {\dellipse};
				\fill[red!40] \rellipse;
			\end{scope}

			\draw (0,4.8) node {\scalebox{1}{$\textsf{$\Rel$-rational}$}};

			\draw(0,3.6) node  {\scalebox{1}{$\textsf{$\Rel$-polyregular}$}};
			\draw(-1.85,3.25) node {\textcolor{red}{\scalebox{1.3}{$\substack{\text{\textsf{Star-free}}\\
								\text{\textsf{$\Rel$-polyregular}}}$}}};

			\draw(-1,0.04) node {\textcolor{red}{\scalebox{1}{$\We_0$}}};
			\draw(-1.4,1.1) node {\textcolor{red}{\scalebox{1}{$\We_1$}}};
			\draw(-1.7,2.2) node {\textcolor{red}{\scalebox{1}{$\We_2$}}};

			\draw(0,0) node {\scalebox{1}{$\Ve_0$}};
			\draw(0,1.2) node {\scalebox{1}{$\Ve_1$}};
			\draw(0,2.4) node {\scalebox{1}{$\Ve_2$}};

			\draw(0,4.1) node  {\scalebox{0.6}{$\textsf{Polynomial~growth}$}};
			\draw(0,2.9) node  {\scalebox{0.6}{$\bigO(n^2) \textsf{~growth}$}};
			\draw(0,1.7) node  {\scalebox{0.6}{$\bigO(n) \textsf{~growth}$}};
			\draw(0,0.5) node  {\scalebox{0.6}{$\bigO(1) \textsf{~growth}$}};

			\fill[darkgray] {(0.9,-0.1) circle (0.05)};
			\draw[darkgray,dotted, thick] (0.9,-0.1) -- (3.85,-0.1);
			\node[right] at (3.9,-0.1) {\scalebox{1}{\textcolor{darkgray}{%
                            \textsf{
						$w \mapsto \mathbf{1}_L(w)$ if $L$ is regular but not
        star-free}}}};

			\fill[darkgray] {(1.2,1) circle (0.05)};
			\draw[darkgray,dotted, thick] (1.2,1) -- (3.85,1);
			\node[right] at (3.9,1) {\scalebox{1}{\textcolor{darkgray}{$w \mapsto |w|\times (-1)^{|w|} $}}};

			\fill[darkgray] {(-1.1,-0.3) circle (0.05)};
			\draw[darkgray,dotted, thick] (-1.1,-0.3) -- (-3.85,-0.3);
			\node[left] at (-3.9,-0.3) {\scalebox{1}{\textcolor{darkgray}{$w \mapsto \mathbf{1}_L(w)$ \textsf{if} $L$ \textsf{star-free}}}};

			\fill[darkgray] {(-2.5,1.5) circle (0.05)};
			\draw[darkgray,dotted, thick] (-2.5,1.5) -- (-3.85,1.5);
			\node[left] at (-3.9,1.5) {\scalebox{1}{\textcolor{darkgray}{$w \mapsto |w|_a \times |w|_b$ \textsf{if} $a,b \in A$ }}};

			\fill[darkgray] {(2.5,4.5) circle (0.05)};
			\draw[darkgray,dotted, thick] (2.5,4.5) -- (3.85,4.5);
			\node[right] at (3.9,4.55) {\scalebox{1}{\textcolor{darkgray}{$w
							\mapsto (-2)^{|w|}$}}};

		\end{tikzpicture}
    }
	\end{center}

	\caption{\label{fig:Z-classes} The classes of functions
    studied in this paper.}
\end{figure*}
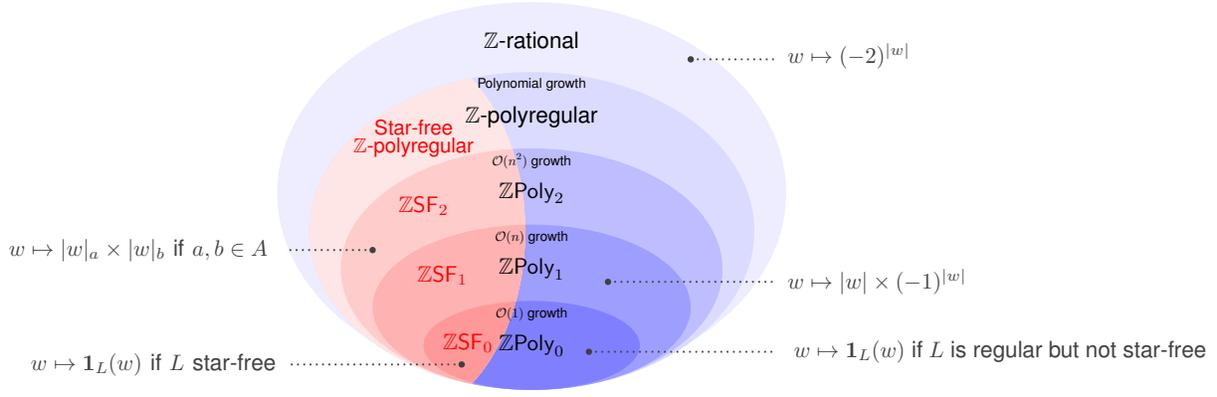

%% file: parts/prelim.tex
\section{$\Rel$-polyregular functions}

\label{sec:prelim}

The goal of this section is to define
$\Rel$-polyregular functions.
We first define this class of functions using a logical formalism
(monadic second-order formulas with free variables,
\cref{ssec:counting}), then we relate it to (word-to-word)
regular and polyregular functions (\cref{ssec:pebble})
and finally we show that it corresponds to a natural
and robust subclass of the well-known
$\Rel$-rational series (\cref{ssec:matrix,ssec:expressions}).

In the rest of this paper, $\Rel$ (resp. $\Nat$) denotes
the set of integers (resp. nonnegative integers).
If $ i \le j$, the set $[i{:}j]$ is $\{i, i{+}1, \dots, j\} \subseteq \Nat$
(empty if $j <i$). The capital letter $A$ denotes a fixed alphabet,
i.e. a finite set of letters. $A^*$ (resp. $A^+$) is the set of words
(resp. non-empty words) over $A$.
The empty word is $\movi \in A^*$.
If $w \in A^*$, let $|w| \in \Nat$ be its length,
and for $1 \le i \le |w|$ let $w[i]$ be its $i$-th letter.
If $I = \{i_1 < \cdots < i_{\ell}\} \subseteq [1{:}|w|]$,
let $w[I] \defined w[i_1] \cdots w[i_{\ell}]$.
If $a \in A$, let $|w|_a$ be the number of letters $a$ occurring in $w$.
We assume that the reader is familiar with the basics of automata theory,
in particular the notions of monoid morphisms, idempotents in monoids,
monadic second-order
($\MSO$) logic and first-order ($\FO$) logic over finite words
(see e.g. \cite{thomas1997languages}).

\subsection{Counting valuations on finite words}

\label{ssec:counting}

Let $\MSO_k$ 
be the set of $\MSO$-formulas  
over the signature $(A, <)$ which have
exactly $k$ free first-order variables.
We then let $\MSO{}\defined \bigcup_{k \in \Nat} \MSO_k$.
If $\phi(x_1, \dots, x_k) \in \MSO_k$,
$w \in A^*$ and $1 \le i_1, \dots, i_k \le |w|$, we
write $w \models \phi(i_1, \dots, i_k)$ whenever
the valuation $x_1 \mapsto i_1, \dots,
x_k \mapsto i_k$ makes the formula $\phi$
true in the model $w$.

\begin{definition}[Counting]
	Given $\phi(x_1, \ldots, x_k) \in \MSO_k$, we let $\card \phi \colon A^* \fonc \Nat$
	be the function defined by $\card \phi(w)
	\defined |\setof{(i_1, \dots,i_k)}{w \models \phi(i_1, \dots, i_k)}|$.
\end{definition}

The value $\card \phi(w)$ is the number of tuples
that make the formula $\phi$ true in the model $w$.

\begin{example} If $\phi \in \MSO_0$, then
$\card \phi$ is the indicator function of the (regular)
language $\setof{w}{w \models \phi} \subseteq A^*$.
\end{example}

\begin{example}
	\label{ex:a-b}
	Let $A \defined \{a,b\}$. Let $\phi(x,y) \defined  a(x) \land b(y)$,
	then $\card \phi (w) = |w|_{a} \times |w|_{b}$ for all $w \in A^*$.
	Let $\psi(x,y) \defined \phi(x,y) \land x > y$,
	then $\card \psi (a^{n_0}b a^{n_1} \cdots a^{n_{p}}) = \sum_{i=0}^p i \times n_i$.
\end{example}

\begin{example}
    \label{ex:times-w}
    Let $\phi \in \MSO_k$, 
    and $x$ be a fresh variable.
    Then, $x = x \land \phi \in \MSO_{k+1}$, and
    $\card{(x = x \land \phi)}(w) = |w| \times \card{\phi}(w)$
    for every
    $w \in A^*$.
    Similarly, for all $w \in A^*$ and $a \in A$,
    $\card (a(x) \land \phi)(w) = |w|_a \times \card{\phi}(w)$.
\end{example}

If $F$ is a subset of the set of functions
$A^* \fonc \Rel$ and if $S \subseteq \Rel$, we let
$\intro\Vect{S}{{{}}}(F) \defined \setof{\sum_i {a_i} f_i}{a_i \in S, f_i \in F}$
be the set of \intro{$S$-linear combinations} of the functions from $F$.
The set $\Vect{\Nat}{{}}(\setof{\card \phi}{\phi \in \MSO_k, k \ge 0})$
has been recently studied by Douéneau-Tabot
in \cite{doueneau2022hiding}
under the name of ``polyregular functions with unary output''.
In the following, we
shall call this class the \emph{$\Nat$-polyregular functions}.

The goal of this paper is to study the
$\Rel$-linear combinations of the basic
$\card \phi$ functions, which we call
\emph{$\Rel$-polyregular functions}. We 
shall see that this class
is a quantitative counterpart of regular languages
that admits several equivalent descriptions,
and for which various decision problems can be solved.
We provide in \cref{def:Z-poly} a fine-grained
definition of this class of functions, depending
on the number of free variables which are used
within the $\card \phi$ basic functions.

\begin{definition}[$\Rel$-polyregular functions]
	\label{def:Z-poly}
	For $k \ge 0$, let
	$\Ve_k \defined \Vect{\Rel}{{}}(\setof{\card \phi}{\phi \in \MSO_\ell, \ell
    \leq k})$.
	We define the class of \emph{$\Rel$-polyregular functions}
	as $\Ve \defined \bigcup_k \Ve_k$.
\end{definition}

We also let $\Ve_{-1} \defined \{0\}$.

\begin{remark}
    \label{rem:qmso-nul}
    For all $k \ge 0$,
    the class $\Ve_k$
    is precisely the class
    of functions computable in
    $\mathsf{QMSO}(\Sigma^k_x, \oplus, \odot_b)$ of
    \cite[Section IV.A]{kreutzer2013}
    over the semiring $(\Rel,+,\times)$.
\end{remark}

\begin{remark}
    \label{rem:sch-nul}
    For all $k \geq 0$,
    the class $\Ve_k$
    is precisely the class of
    functions computable
    by \emph{linear
    finite 
    counting automata of order $k$}
    introduced by
    \cite[p. 91]{schutzenberger1962}.
\end{remark}

\begin{example}
    \label{ex:Ve0}
    $\Ve_0$ is exactly the class of 
    $\Rel$-linear combinations of
    indicators
    $\mathbf{1}_{L}$
    of regular languages $L$.
\end{example}

\begin{example}
    \label{ex:times-w-k}
    Following the construction of \cref{ex:times-w},
    for every $k, \ell \geq 0$, and $f \in \Ve_\ell$,
    the function
    $g \colon  w \mapsto f(w) \times |w|^k$
    belongs to $\Ve_{\ell + k}$.
\end{example}

\def\Leven{\car{\textnormal{even}}}
\def\Lodd{\car{\textnormal{odd}}}

\begin{example} 
    Let $\car{\textnormal{odd}}$
    and $\car{\textnormal{even}}$ be respectively the indicator functions
    of words of odd length and even length.
    For all $k \ge 0$, the function
	$w \mapsto (-1)^{|w|} \times |w|^k$ is in $\Ve_k$.
    Indeed, it is
    $w \mapsto \car{\textnormal{even}}(w) \times  |w|^k
    - \car{\textnormal{odd}}(w) \times |w|^k$.
	Observe that it cannot be written as a single $\delta \card \phi$
	for some $\delta \in \Rel$, $\phi \in \MSO_{\ell}$, $\ell \ge 0$,
	since otherwise its sign would be constant.
\end{example}

The use of negative coefficients in the linear
combinations has deep consequences on the
expressive power of $\Ve$.
Let us consider
the function $f \colon w \mapsto (|w|_a - |w|_b)^2$.
Because $f(w) = |w|_a^2 - 2|w|_a |w|_b + |w|_b^2$,
we conclude from \cref{ex:times-w} that
$f$ is in $\Ve_2$. Although $f$ is non-negative,
$f^{-1}(\{0\}) = \setof{w}{|w|_a = |w|_b}$ 
is not a regular language,
hence
$f$ is not a \kl{$\Nat$-polyregular function}.

\begin{remark}[More variables]
    \label{remark:zerology}
    Let $ \ell > k \ge 0$, $\phi \in \MSO_k$, then
    for all words $w \in A^+$ we have:
    \begin{equation*}
    \card{\phi}(w) = \card{(\phi \land x_{k+1} = \cdots = x_{\ell}
    \land \forall y. x_{k+1} \le y)}(w)
    \end{equation*}
    the latter being an $\MSO_\ell$ formula.
    This formula also holds for $w = \movi$ if $k > 0$,
    but it may fail for $k = 0$ because in that case
    the right-hand side equals $0$
    regardless of the formula $\phi$
    (because there is no valuation),
    whereas $\card \varphi(\movi)$ may not be $0$.
\end{remark}

One can refine \cref{remark:zerology} to conclude
that for all $k \geq 0$,
$\Ve_k = 
\Vect{\Rel}{{}}(\setof{\card \phi}{\phi \in \MSO_k} \cup \set{ \car{\set{\movi}}})$.
In the rest of the paper, $\car{\set{\movi}}$ will not play any role,
and we will safely ignore it in the proofs
so that
$\Ve_k$ will often be considered equal to
$\Vect{\Rel}{{}}(\setof{\card \phi}{\phi \in \MSO_k})$.

%% file: parts/Zrat.tex
\subsection{Regular and polyregular functions}
\label{ssec:pebble}

We recall that the class of (word-to-word) functions computed by 
\emph{two-way transducers} (or equivalently 
by $\MSO$-transductions, see e.g. \cite{engelfriet2001mso})
is called \label{def:regular}\emph{regular functions}.
As an easy consequence of its definition,
$\Ve_k$ is preserved under pre-composition with a regular function.

\begin{proposition}
	\label{prop:precompose}
    For all $k \geq 0$, the class $\Ve_k$ is (effectively) closed under
    pre-composition by regular functions.
\end{proposition}

Now, we intend to justify the name
``$\Rel$-polyregular functions'' by showing
that this class is deeply connected to
the well-studied class of \emph{polyregular functions}
from finite words to finite words.
Informally, this class of functions can be defined using
the formalism of multidimensional $\MSO$-interpretations.
The reader is invited
to consult \cite{bojanczyk2019string}  for its formal
definition, which we skip here.
Let $\polysum : \{\pm 1\}^* \fonc \Rel$ be the sum operation
mapping $w \in  \{\pm 1\}^*$ to $\sum_{i=1}^{|w|} w[i]$.

\begin{proposition}\label{prop:polypoly}
    \label{prop:mikolaj1}
	\label{prop:mikolaj1bis}
	The class $\Ve$ is (effectively) the class
	of functions $\polysum \circ f$ where $f : A^* \fonc \{\pm1\}^*$
	is polyregular.
\end{proposition}

\subsection{Rational series and rational expressions}
\label{ssec:expressions}

The class of rational series over the semiring
$(\Rel, +, \times)$, also known as \emph{$\Rel$-rational series},
is a robust class of functions from finite words to $\Rel$
that has been largely studied since the 1960
(see e.g. \cite{berstel2011noncommutative} for a survey).
It can be defined using the indicator functions
$\car{L}$ of regular languages
$L \subseteq A^*$, and the following combinators
given $f,g: A^* \rightarrow \Rel$ and $\delta \in \Rel$:
\begin{itemize}
	\item the external $\Rel$-product $\delta f \colon w \mapsto \delta \times f(w)$;
	\item the \intro{sum} $f+g: w \mapsto f(w) + g(w)$;
	\item the \intro{Cauchy product} $ f \cauchy g \colon w \mapsto \sum_{w = uv} f(u) \times g(v)$;
	\item if and only if $f(\movi) = 0$, the \intro{Kleene star}
	      $ f^* \defined \sum_{n \ge 0} f^n$ where
	      $f^{0}: \movi \mapsto 1, w \neq \movi \mapsto 0$
	      is neutral for Cauchy product
	      and $f^{n+1} \defined f \cauchy f^{n}$.
\end{itemize}

\begin{definition}[$\Rel$-rational series]
    The class of \emph{$\Rel$-rational series} is the smallest
	class of functions from finite words to $\Rel$
	that contains the indicator functions of all regular languages,
	and is closed under taking external $\Rel$-products, sums,
	Cauchy products and Kleene stars.
\end{definition}

We intend to connect  $\Rel$-rational
series and $\Rel$-polyregular functions. Let us first
observe that not all  $\Rel$-rational series
are $\Rel$-polyregular. We say that a function $f : A^* \fonc \Rel$ has
\emph{polynomial growth} whenever there exists
$k \ge 0$ such that $|f(w)| = \bigO(|w|^k)$.
It is an easy check that
every $\Rel$-polyregular function has polynomial growth.

\begin{claim}
    \label{remark:z-poly-polygrowth}
    If $k \ge 0$ and $f \in \Ve_k$ then $|f(w)| = \bigO(|w|^k)$.
\end{claim}

\begin{example} 
    \label{ex:moins-deux-puissance-n}
    The map $f \colon w \mapsto (-2)^{|w|}$
    is a $\Rel$-rational series
    because $f = ((-3) \car{A^+})^*$.
    However $f \not \in \Ve$ since
    it does not have polynomial growth.
\end{example}

It is easy to see from the logical definition
that the class $\Ve$ is closed under
taking Cauchy products.

\begin{claim}
	\label{claim:vks:cauchyincrease}
	Let $k, \ell \ge 0$. Let $f \in \Ve_{k}$ 
	and $g \in \Ve_{\ell}$, then
	$f \cauchy g \in \Ve_{k+\ell+1}$.
	The construction is effective.
\end{claim}

As a consequence, if $L \subseteq A^*$ is regular
and $f \in \Ve_k$, then $\car{L} \cauchy f \in \Ve_{k+1}$.
The following result states that such functions
actually generate the whole space $\Ve_{k+1}$.

\begin{proposition}
	\label{lem:vks:inductcauchy}
	Let $k \geq 0$, the following (effectively)
	holds:
	$$\Ve_{k+1} = \Vect\Rel{{}}(\setof{\mathbf{1}_{L} \cauchy f}{L \text{\normalfont~regular}, f \in \Ve_k}).$$
\end{proposition}

\begin{example}
    \label{ex:minus-one-times-size}
	The map $w \mapsto (-1)^{|w|} |w|$
    is in $\Ve_1$ as it equals
    $\Lodd \cauchy \Lodd + \Leven \cauchy \Leven
		- \Leven \cauchy \Lodd - \Lodd \cauchy \Leven
    - \Lodd + \Leven 
        $.
\end{example}

Now, let us show that $\Rel$-polyregular functions
can be characterized both syntactically and semantically
as a subclass of $\Rel$-rational series. We
prove that the membership problem is decidable and provide
an effective conversion algorithm.

\begin{theorem}[Rational series of polynomial growth]
    \label{theo:rational-polygrowth}
    Let $f:A^* \fonc \Rel$, the following are equivalent:
    \begin{enumerate}
    \item \label{it:popoly} $f$ is a $\Rel$-polyregular function;
    \item \label{it:cauchy}
    $f$ belongs to the smallest class of functions that
    contains the indicator functions of all regular languages
    and is closed under taking external $\Rel$-products, sums and
    Cauchy products;
    \item \label{it:ratpoly}
    $f$ is a $\Rel$-rational series having polynomial growth.
    \end{enumerate}
    Furthermore, one can decide whether a $\Rel$-rational series
    is a $\Rel$-polyregular function and the translations
    are effective.
\end{theorem}

\begin{proof}

	For \cref{it:cauchy} $\implies$ \cref{it:popoly},
	observe that $\Ve$  contains the indicator functions of regular languages, is closed 
	under external $\Rel$-products, sums,
	and Cauchy products (thanks to 
	\cref{claim:vks:cauchyincrease}). For \cref{it:popoly}
	$\implies $ \cref{it:cauchy}, we obtain for all $k \ge 0$ as an
	immediate consequence of \cref{lem:vks:inductcauchy}:
    \begin{equation}
        \label{eq:cauchys}
        \begin{array}{rcll}
             & \Ve_k & = \Vect{\Rel}(\{\car{L_0} \cauchy \cdots \cauchy
            \car{L_k}& \\ &&\colon L_0, \dots, L_k \text{ regular languages}\})
        \end{array}
    \end{equation}
    and the result follows.
    
	The equivalence between \cref{it:cauchy}
	and \cref{it:ratpoly} follows (in a non effective way)
	from \cite[Corollary 2.6 p 159]{berstel2011noncommutative}.
	Furthermore polynomial growth is decidable by
	\cite[Corollary 2.4 p 159]{berstel2011noncommutative}.
	To provide an effective translation, one can start from a
	$\Rel$-rational series $f$ of polynomial growth,
	enumerate all the $\Rel$-polyregular functions $g$,
    rewrite them as rational series (using \cref{it:popoly} $\implies$ \cref{it:cauchy})
    and check whether $f = g$ since this property
    can be decided for $\Rel$-rational series
    \cite[Corollary 3.6 p 38]{berstel2011noncommutative}.
\end{proof}

\begin{remark}
    It follows from \cref{rem:qmso-nul},
    \cite[Proposition 6.1]{kreutzer2013},
    and 
    \cref{theo:rational-polygrowth},
    that $\Rel$-rational series
    of polynomial growth are exactly
    those computable by
    weigthed automata with
    coefficients in $\{0,1,-1\}$
    of polynomial ambiguity.
    We are not aware of a direct
    proof of this correspondence.
\end{remark}

\begin{remark}
\cite[Theorem 3.3]{doueneau2022hiding}
gives a similar result when comparing $\Nat$-polyregular functions
and $\Nat$-rational series.
\end{remark}

\begin{remark}
    The class of $\Rel$-polyregular functions
	is also closed under Hadamard product
	($f \times g(w) \defined f(w) \times g(w)$).
    This can be obtained by generalising
    \cref{ex:times-w}.
    Moreover, $f \times g \in \Ve_{k+\ell}$
    whenever $f \in \Ve_{k}$ and $g \in \Ve_{\ell}$.
\end{remark}

Since the equivalence is decidable
for $\Rel$-rational series
\cite[Corollary 3.6 p 38]{berstel2011noncommutative},
we obtain the following.

\begin{corollary}[Equivalence problem]
\label{cor:equivalence}
One can decide if two $\Rel$-polyregular functions
are equal.
\end{corollary}


%
%
%

\subsection{Rational series and representations}
\label{ssec:matrix}

In this section, we intend to provide another
description of $\Rel$-polyregular functions among $\Rel$-rational series.
To that end, we first recall that rational series
can also be described using matrices (or, equivalently,
weighted automata). Let $\Mat{n,m}(\Rel)$ be the set of all
$n \times m$ matrices with coefficients in $\Rel$. We equip
$\Mat{n,m}(\Rel)$ with the usual matrix multiplication.

\begin{definition}[Linear representation]
	We say that a triple $(I,\mu, F)$ where
	$\mu \colon A^* \to \Mat{n,n}(\Rel)$ is a monoid morphism,
	$I \in \Mat{1,n}(\Rel)$ and $F \in \Mat{n,1}(\Rel)$,
	 is a \emph{$\Rel$-linear representation} of a function $f: A^* \fonc \Rel$
	 if  $f(w) = I \mu(w) F$ for all $w \in A^*$.
\end{definition}

It is well-known since Schützenberger (see e.g.
\cite[Theorem~7.1 p 17]{berstel2011noncommutative})
that the class of $\Rel$-rational series is (effectively) the class
of functions that have a $\Rel$-linear representation.

\begin{example}
    The map $w \mapsto (-1)^{|w|} |w|$
    from \cref{ex:minus-one-times-size}
    is a $\Rel$-polyregular function, hence
    it is a $\Rel$-rational series. It has the following
    $\Rel$-linear representation:
    \begin{equation*}
        \left(\begin{pmatrix} -1 & 0 \end{pmatrix},
        w \mapsto \begin{pmatrix} -1 & 1 \\ 0 & -1 \end{pmatrix}^{|w|},
        \begin{pmatrix} 0 \\ 1 \end{pmatrix}
        \right).
    \end{equation*}
    Note that the eigenvalues of any matrix in $\mu(A^*)$
    are $1$ or $-1$. 
\end{example}

\begin{example}
    The function $w \mapsto (-2)^{|w|}$
    from \cref{ex:moins-deux-puissance-n}
   is a $\Rel$-rational series that is not a $\Rel$-polyregular function.
    It can be represented via $((1),\mu,(1))$ where
	$\mu(w) = ((-2)^{|w|})$ for all $w\in A^*$.
	Observe that for all $n \ge 1$, there exists
	a matrix in $\mu(A^*)$ whose eigenvalue has modulus
	$2^n > 1$.
\end{example}

A $\Rel$-linear representation $(I,\mu,F)$ of a function $f$
is said to be
\emph{minimal}, when it has minimal dimension $n$ among
all the possible representations of $f$.
Given a matrix $M \in \Mat{n,n}(\Rel)$, we let $\Spec(M) \subseteq \Com$
be its \emph{spectrum}, which is the set of
all its (complex) \emph{eigenvalues}.
If $S \subseteq \Mat{n,n}(\Rel)$,
we let $\Spec(S) \defined \bigcup_{M \in S} \Spec(M)$
be the union of the spectrums.
Finally, let $B(0,1) \defined \setof{x \in \Com}{|x| \le 1}$
be the unit disc and
$\Uni \defined \setof{x \in \Com}{\exists n \ge 1, x^n =1}$
be the roots of unity.

Now, we show that $\Rel$-polyregular functions can
be characterized through the eigenvalues of $\Rel$-linear representations.
More precisely,  \cref{theo:Zpoly-Zrat}
will relate the asymptotic growth of a series to the
spectrum of the set of matrices $\mu(A^*)$.  As a first step,
let us observe that the eigenvalues occurring in a minimal representation
can be revealed by iterating words.

\begin{lemma}
    \label{lem:capturing-eigenvalues}
    Let $f \colon A^* \to \Rel$ be a $\Rel$-rational series and
	$(I, \mu, F)$ be a minimal $\Rel$-linear representation of $f$.
	Let $w \in A^*$ and $\lambda \in \Spec(\mu(w))$.
    There exist coefficients 
    $\alpha_{i,j} \in \mathbb{C}$ 
    for $1 \leq i,j \leq n$, 
    and words $u_1, v_1, \dots, u_n, v_n \in A^*$ such that
    $\lambda^X = \sum_{i ,j= 1}^{n} \alpha_{i,j}  f(v_i w^X u_j)$
    for all $X \ge 0$.
\end{lemma}

Now, we refine the notion of polynomial
growth to explicit the behaviour of
a function when iterating factors.

\begin{definition}
	\label{def:ultimately-polynomial-general}
	Let $N > 0$. A function $f \colon A^* \fonc \Rel$ is 
    \emph{ultimately $N$-polynomial}
    whenever there exists
    $M \ge 0$ such that
    for all $\ell \ge 0$, for all
    $\alpha_0, w_1, \alpha_1, \dots, w_{\ell}, \alpha_{\ell} \in A^*$,
    there exists $P \in \Rat[X_1, \dots, X_\ell]$,
	such that
    $f(\alpha_0 w_1^{N X_1} \alpha_1 \cdots w_{\ell}^{N X_{\ell}}\alpha_{\ell})
    = P(X_1, \dots, X_\ell)$,
    whenever $X_1, \dots, X_\ell \geq M$.
\end{definition}

In this section we only need to have
$\ell = 1$, but \cref{def:ultimately-polynomial-general} has been made generic
so that it can be reused in \cref{sec:aperiodic} when dealing
with aperiodicity. Now, we observe that ultimate polynomiality
is preserved under taking sums, external $\Rel$-products
and Cauchy products. \Cref{lem:Zpoly-pump-first} also provides
a fine-grained control over the value $N$ of ultimate $N$-polynomiality,
that will mostly be useful in \cref{sec:aperiodic}.

\begin{lemma}
    \label{lem:Zpoly-pump-first}
    Let $f, g \colon A^* \to \Rel$ be (respectively) ultimately
    $N_1$-polynomial
    and ultimately $N_2$-polynomial, then:
    \begin{itemize}
    \item $f + g$ and $f \cauchy g$
    are ultimately $(N_1 \times N_2)$-polynomial;
    \item $\delta f$ is ultimately $N_1$-polynomial
    for $\delta \in \Rel$.
    \end{itemize}
    Furthermore, for every
    regular language $L$, there exists $N > 0$
    such that $\car{L}$ is ultimately $N$-polynomial.
\end{lemma}

Now, we have all the elements to prove the main theorem of this section.

\begin{theorem}[Polynomial growth and eigenvalues]
	\label{theo:Zpoly-Zrat}
	Let $f: A^* \fonc \Rel$, the following are equivalent:
	\item
	\begin{enumerate}
		\item \label{it:polyreg} $f$ is a $\Rel$-polyregular function;
		\item \label{it:Zpoly-pump-first} $f$ is a $\Rel$-rational series
		that is ultimately $N$-polynomial
		for some $N > 0$;
		\item \label{it:unity} $f$ is a $\Rel$-rational series
		and for all minimal $\Rel$-linear representations $(I, \mu,F)$ of $f$,
		$\Spec(\mu(A^*)) \subseteq \Uni \cup \{0\}$.
		\item \label{it:disc} $f$ is a $\Rel$-rational series
		and for some minimal $\Rel$-linear representation $(I, \mu,F)$ of $f$,
		$\Spec(\mu(A^*)) \subseteq B(0,1)$;
	\end{enumerate}
\end{theorem}

\begin{proof}
	\Cref{it:disc} $\implies$ \cref{it:polyreg}
	is a direct consequence of \cite[Theorem 2.6]{bell_gap_2005}
	and \cref{theo:rational-polygrowth}.
	\Cref{it:polyreg} $\implies$ \cref{it:Zpoly-pump-first}
   	 follows from \cref{lem:Zpoly-pump-first} and \cref{theo:rational-polygrowth}.
	
	For \cref{it:Zpoly-pump-first} $\implies$ \cref{it:unity},
	let $(I, \mu, F)$ be a minimal representation of $f$ in $\Rel$,
	of dimension $n \ge 0$.
	Let $w \in A^*$ and $\lambda \in \Spec(\mu(w))$.
	Thanks to \cref{lem:capturing-eigenvalues},
    there exists $\alpha_{i,j}, u_i, v_j$ for $1 \leq i,j \leq n$, such that
    $\lambda^X = \sum_{1 \leq i,j \leq n} \alpha_{i,j}  f(v_i w^X u_j)$
	 for $X$ large enough. By assumption, for all $1 \le i,j \le n$, there exists
	$N_{i,j} > 0$ such that $X \mapsto f(v_i w^{N_{i,j} X} u_j)$
	is a polynomial for $X$ large enough. Hence there exists $N > 0$
	(i.e. the product of the $N_{i,j}$) such that $X \mapsto \lambda^{NX} = (\lambda^N)^X$
	is a polynomial for $X$ large enough,
	which therefore must be a constant polynomial.
	Hence $\lambda^N \in \{0,1\}$, which implies
	that $\lambda \in \{0\} \cup \mathbb{U}$.
	\Cref{it:unity} $\implies$ \cref{it:disc} is obvious.
\end{proof}

\begin{remark} \Cref{it:unity} of \cref{theo:Zpoly-Zrat}
is optimal, in the sense that for all $\lambda \in \Uni \cup \{0\}$,
there exists a $\Rel$-rational series of polynomial growth
having a minimal representation $(I, \mu, F)$ with $\lambda \in
\Spec(\mu(A^*))$ (if $\lambda \in \Uni$, we let $\mu(a)$ be
the companion matrix of the cyclotomic polynomial associated to $\lambda$).
\end{remark}

\begin{remark}
    \label{remark:N-polynomial-monomial}
    Leveraging the proof scheme used 
    for the
    implication \cref{it:Zpoly-pump-first} $\implies$ \cref{it:unity}
    of
    \cref{theo:Zpoly-Zrat},
    one can
    actually show 
    that the following asymptotic polynomial bound
    characterizes
    $\Rel$-polyregular functions among $\Rel$-rational series:
    for all  $u,w,v \in A^*$,
    there exists $P \in \Rat[X]$,
	such that
    $|f(u w^{X} v)| \leq P(X)$,
    for $X$ large enough. 
\end{remark}

\begin{remark} Beware that $\Spec(\mu(A)) \subseteq \set{0} \cup \Uni$
has no reason to imply $\Spec(\mu(A^*)) \subseteq \set{0} \cup \Uni$.
\end{remark}

%% file: parts/esqueletto.tex
\section{Free Variable Minimization and Growth Rate}
\label{sec:pebblemin}

In this section, we study the membership problem
from $\Ve$ to $\Ve_k$ for a given $k \ge 0$.
As observed in \cref{remark:z-poly-polygrowth},
if $f \in \Ve_k$ then
$|f(w)| = \bigO(|w|^k)$. We show 
that this asymptotic behavior completely
characterizes $\Ve_k$ inside $\Ve$.
This statement is formalized in \cref{thm:skel:pebblemin},
which also provides both a decision procedure
and an effective conversion algorithm.
It turns out that \cref{thm:skel:pebblemin} is 
also a stepping stone towards computing
the \kl{residual automaton} of a function  $f \in \Ve$,
which is done in \cref{sec:residual}.

This can be understood as a result
that ``minimizes" the number of free variables needed
to describe a $\Rel$-polyregular function. 
As such, it
is tightly connected with the ``pebble minimization''
results that exists
for (word-to-word) polyregular functions
\cite{bojanczyk2022growth}
and $\Nat$-polyregular functions 
\cite{doueneau2021pebble}.
However, these results cannot be used as black box theorems
to minimize the number of free variables of $\Rel$-polyregular
functions because the negative coefficients of the latter
induce non-trivial behaviors.

To capture the growth rate of $\Rel$-polyregular functions,
we shall introduce a 
quantitative variant of the traditional
pumping lemmas. Before that, let us
extend the $\bigO$ notation to multivariate functions
$f,g : \Nat^n \fonc \Rel$ as follows:
we say that $f= \bigO(g)$ whenever
there exist $N,C \ge 0$ such that
$|f(x_1, \dots, x_n)| \leq C |g(x_1,\dots, x_n)|$ for every
$x_1, \dots, x_n \geq N$.
We similarly extend the notation $f(x) = \Omega(g(x))$ to multivariate
functions. 

\begin{definition}
    \label{def:k-pumpable}
    A function $f \colon A^* \to \Rel$ is
    \emph{$k$-pumpable} whenever there exist
    $\alpha_0, \dots, \alpha_k \in A^*$,
    $w_1, \dots, w_k \in A^*$, 
    such that
    $|f(\alpha_0 \prod_{i = 1}^k w_i^{X_i} \alpha_i)|
    =
    \Omega(|X_1 + \cdots + X_k|^k)$.
\end{definition}

\begin{example}
    For all $k \geq 0$, for all $f \in \Ve_{k}$, $f$ is not \kl{$(k+1)$-pumpable}
    because $|f(w)| = \bigO(|w|^k)$.
\end{example}

The equivalence between \cref{item:fvmin-f-ve}
and \cref{item:fvmin-f-polyk} in \cref{thm:skel:pebblemin}
is known since \cite{schutzenberger1962}. However,
the equivalence with \cref{item:fvmin-f-pump}
and the effectivity of the result are novel.

\begin{restatable}[Free Variable Minimization]{theorem}{pebbleminimisation}
	\label{thm:skel:pebblemin}
	Let $f \in \Ve$ and $k \ge 0$. The following conditions are equivalent:
    \begin{enumerate}
        \item \label{item:fvmin-f-ve} $f \in \Ve_k$;
        \item \label{item:fvmin-f-polyk} $|f(w)| = \bigO(|w|^k)$;
        \item \label{item:fvmin-f-pump} $f$ is not $(k+1)$-pumpable.
    \end{enumerate}
	Furthermore, the minimal $k$
    such that $f \in \Ve_k$ is computable,
	and the construction is effective.
\end{restatable}

The proof of \cref{thm:skel:pebblemin} is done via induction on $k$,
and follows directly from the following induction step, to which 
we devote the rest of \cref{sec:pebblemin}.


\begin{inductionstep}
	\label{lem:skel:semantic-cond}
	Let $k \geq 1$ and $f \in \Ve_k$. The following conditions are equivalent:
	\begin{enumerate}
		\item \label{it:cond:class} $f \in \Ve_{k-1}$;
		\item \label{it:cond:growth} $|f(w)| = \bigO(|w|^{k-1})$;
		\item \label{it:cond:tuple}
            $f$ is not \kl{$k$-pumpable}.
	\end{enumerate}
	Moreover this property can be decided and
	the construction is effective.
\end{inductionstep}

Beware that one must be able to pump several factors at once
to detect the growth rate, as illustrated in the following example.
This has to be contrasted with \cref{remark:N-polynomial-monomial}.

\begin{example}
    \label{ex:skel:multi-iter}
	Let $f: a^kb^{\ell} \mapsto k \times \ell$ and $w \mapsto 0$ otherwise.
	The function $f$ is $\Rel$-polyregular and $2$-pumpable,
	however, $f(\alpha_0 w^X \alpha_1) = \bigO(X)$ for every triple $\alpha_0,w,\alpha_1 \in A^*$.
\end{example}

Our proof of \cref{lem:skel:semantic-cond} is built upon
factorization forests.
Given a morphism $\mu \colon A^* \to M$ into a finite monoid and
$w \in A^*$, a \kl{$\mu$-forest of $w$} is a forest that can be represented
as a word over $\hat{A} \defined A \uplus \{ \langle, \rangle \}$, defined
as follows.

\begin{definition}[Factorization forest \cite{simon1990factorization}]
	Given a monoid morphism $\mu \colon A^* \to M$ and $w \in A^*$, we say that
	$F$ is a \intro{$\mu$-\emph{forest}} of $w$ when:
	\begin{itemize}
		\item either $F = a$, and $w = a \in A $;
		\item or $F = \tree{F_1} \cdots \tree{F_n}$,
		      $w = w_1 \cdots w_n$ and for all $1 \le i \le n$,
		      $F_i$ is a $\mu$-forest of $w_i \in A^+$.
		      Furthermore, if $n \ge 3$ then
		      $\mu(w_1) = \dots = \mu(w_n)$ is an idempotent of $M$.
	\end{itemize}
\end{definition}

We let $\forests{\mu}{}$ be the language 
of $\mu$-forests inside $(\hat{A})^*$.
Because forests are (ordered) trees, we will use the standard vocabulary
to talk about the nodes, the sibling/parent relation, the root, the leaves
and the depth of a forest. We let $\forests{\mu}{d} \subseteq (\hat{A})^*$ be the set of $\mu$-forests
with depth at most $d$. Let $\repword \colon \forests{\mu}{d} \to A^*$ be the function
mapping a $\mu$-forest of $w  \in A^*$ to $w$ itself. 

\begin{example} \label{ex:forest} Let $M \defined (\{-1,1,0\}, \times)$.
	A forest $F\in \forests{\mu}{5}$
	(where $\mu : M^* \fonc M$ maps a word to the product of its elements)
	such that $\repword(F) = (-1)(-1)0(-1)000000$ is depicted
	in Figure~\ref{fig:ex:skel}. Double lines
	denote idempotent nodes (i.e.
	nodes with more than $3$ children).
\end{example}

When $M$ is a finite monoid, it is known from
Simon's celebrated theorem \cite{simon1990factorization} that
any word in $A^*$ has a $\mu$-forest of
bounded depth. Furthermore, this small forest can be computed
by a regular function (notion introduced in \cref{ssec:pebble}).

\begin{theorem}[\cite{simon1990factorization,colcombet2011green}]
	\label{thm:simon}
	Given a morphism  into a finite monoid $\mu \colon A^* \to M$,
	one can effectively compute some $d \ge 0$ and a 
         regular function $\repforest \colon A^* \to \forests{\mu}{d}$
	such that $\repword \circ \repforest$ is the identity function.
\end{theorem}

In order to prove \cref{lem:skel:semantic-cond},
we shall consider
a function $(f \colon A^* \to \Rel) \in \Ve_k$
that is not $k$-pumpable, and show how to compute it
as a function in $\Ve_{k-1}$.
To that end, we shall construct a
function $g \colon \hat{A}^* \to \Rel \in \Ve_{k-1}$
such that $f = g \circ \repforest$. Since
$\repforest$ is regular thanks to \cref{thm:simon},
it will follow that $f \in \Ve_{k-1}$
by \cref{prop:precompose}.
Remark that it is only needed to define $g$
on $\forests{\mu}{d}$.

Following the classical connections between $\MSO$-formulas
and regular languages \cite{thomas1997languages},
we prove in \cref{lem:skel:to-forest}
that for every function $f \in \Ve_k$ there exist
a finite monoid $M$ and a morphism $\mu \colon A^* \to M$,
such that $f(w)$ can be reconstructed using
``simple'' $\MSO$-formulas which are evaluated
along \kl{bounded-depth $\mu$-factorizations} of $w$.

\begin{claim}
    Given a morphism $\mu \colon A^* \to M$
    into a finite monoid and $d \in \Nat$, 
    the following predicates are $\MSO$ definable
    for words over $\hat{A}$. For all $F \in \forests{\mu}{d}$,
    and $w = \repword(F)$, then:
    \begin{itemize}
        \item
            $F \models \msoisleaf(x)$
            if and only if 
            $x$ is a leaf of $F$;

        \item
            $F \models \msoprod_m(x,y)$
            if and only if 
            $x$ and $y$ are leaves of $F$,
            $x \leq y$, and $\mu(w[x] \dots w[y]) = m$;

        \item
            $F \models \msoleft_m(x)$
            if and only if 
            $x$ is a leaf of $F$,
            and $\mu(w[1] \dots w[x]) = m$;

        \item
            $F \models \msoright_m(x)$
            if and only if 
            $x$ is a leaf of $F$,
            and $\mu(w[x] \dots w[|w|]) = m$.

    \end{itemize}
    Whenever $F \in \hat{A}^* \setminus \forests{\mu}{d}$,
    the semantics are undefined.
\end{claim}

\begin{definition}
    The fragment $\INV$ is a subset of $\MSO$ over $\hat{A}$,
    which contains the quantifier-free formulas using
    only the predicates $\msoprod_m$, $\msoleft_m$, and $\msoright_m$ where
    $m$ ranges over $M$, and where every free variable $x$
    is guarded by the predicate $\msoisleaf(x)$.
    Furthermore, we let $\INV_k \defined \INV \cap \MSO_k$.
\end{definition}

\begin{claim}[{\cite{bojanczyk2022transducers,bojanczyk2022growth}}]
	\label{lem:skel:to-forest}
	For all $f \in \Ve_k$, one can build
	a finite monoid $M$, a depth $d \in \Nat$,
	a surjective morphism $\mu \colon A^* \to M$,
	constants $n \geq 0$,
    $\delta_i \in \Rel$ for $1 \leq i \leq n$,
	formulas $\psi_i \in \INV_k$ for $1 \leq i \leq n$,
	such that
	for every word $w \in A^*$,
	for every \kl{factorization forest} $F \in \forests{\mu}{d}$ of $w$,
    it holds that
	$f(w) = \sum_{i = 1}^{n} \delta_i \times \card{\psi_i}(F)$.
\end{claim}

In the rest of this section, we
focus on the number of free variables in $\Rel$-linear combinations
of $\card \psi$ where $\psi \in \INV$.
The crucial idea is that one can leverage the structure of the forest
$F \in \forests{\mu}{d}$
to compute $\card \psi$ more efficiently, at the cost of
building a non-$\INV$ formula.

For that, we explore the structure of the forest $F$
as follows: given a node $\nod$ in a forest $F$, we define its
skeleton to be the subforest rooted at that node, containing
only the right-most and left-most children recursively.
This notion was already used
in \cite{doueneau2022hiding,doueneau2023pebble,bojanczyk2022growth} for the study
of pebble transducers.

\begin{definition} Let $F \in \forests{\mu}{}$ and $\nod \in \Nodes(F)$, we define
	the \intro{skeleton} of $\nod$ by:
	\begin{itemize}
		\item if $\nod = a\in A$ is a leaf, then $\intro\Ske(\nod) \defined \{ \nod \}$;
		\item otherwise if $\nod = \tree{F_1} \cdots \tree{F_n}$, then
		      $\Ske(\nod) \defined \{\nod\} \cup \Ske(F_1) \cup \Ske(F_n)$.
	\end{itemize}
\end{definition}

Let $w \in A^*$, $F$ be a $\mu$-forest of $w$, and
$\nod \in \Nodes(F)$. The set of nodes $\Ske(\nod)$ defines
a $\mu$-forest of a (scattered) subword $u$ of $w$: the one
obtained by concatenating the leaves of $F$ that are in $\Ske(\nod)$.
See \cref{fig:ex:skel} for an example of a skeleton.
A crucial property of $\Ske(\nod)$ seen as a forest is that
it preserves the evaluation:

\begin{claim}
    \label{claim:skel-semantic-invariant}
    For all $d \geq 0$, finite monoid $M$,
    morphism $\mu \colon A^* \to M$,
    forest $F \in \forests{\mu}{d}$,
    node $\nod \in F$,
    it holds that
    $\mu(\repword(\Ske(\nod))) = \mu(\repword(\nod))$.
\end{claim}

\begin{figure}[h!]
	\newcommand{\couleur}{blue!70}
	\centering
    \resizebox{\linewidth}{!}{
	\begin{tikzpicture}[
			skel/.style={
					circle, fill=blue,
					inner sep=2pt,
					minimum width=1pt,
				}
		]

		\newcommand{\texte}{\small \bfseries \sffamily \mathversion{bold} }

		\node[above] at  (0,0)  {$-1$};
		\node[above] at  (1,0)  {$-1$};
		\node[above] at  (2,0)  {$0$};
		\node[above] at  (3,0)  {$-1$};
		\node[above] at  (4,0)  {$0$};
		\node[above] at  (5,0)  {$0$};
		\node[above] at  (6,0)  {$0$};
		\node[above] at  (7,0)  {$0$};
		\node[above] at  (8,0)  {$0$};
		\node[above] at  (9,0)  {$0$};

		\draw[double] (4,0.975) -- (8,0.975);
		\draw[] (3,1.5) -- (6,1.5);
		\draw[double] (2,1.975) -- (9,1.975);
		\draw[] (0,2) -- (1,2);
		\draw[] (0.5,2.5) -- (6.5,2.5);

		\draw[] (0,2) -- (0,0.5);
		\draw[] (1,2) -- (1,0.5);
		\draw[] (0.5,2) -- (0.5,2.5);

		\draw[] (2,2) -- (2,0.5);
		\draw[] (9,2) -- (9,0.5);

		\draw[] (6.5,2) -- (6.5,2.5);
		\draw[] (3.5,2.5) -- (3.5,2.75);

		\draw[] (3,1.5) -- (3,0.5);
		\draw[] (4,1) -- (4,0.5);
		\draw[] (5,1) -- (5,0.5);
		\draw[] (6,1.5) -- (6,0.5);
		\draw[] (7,1) -- (7,0.5);
		\draw[] (8,1) -- (8,0.5);

		\draw[] (4.5,2) -- (4.5,1.5);

		\node[skel] at (4.5,1.5) {};
		\node[skel] at (6,1) {};
		\node[skel] at (3,0.5) {};
		\node[skel] at (4,0.5) {};
		\node[skel] at (8,0.5) {};

	\end{tikzpicture}
}

	\caption{\label{fig:ex:skel}  A forest $F$ with $\repword(F) = (-1)(-1)0(-1)000000$
		together with a skeleton in blue.}
\end{figure}

Let $F$ be a forest and $x$ be a leaf in $F$.
Observe that $\Ske(x)$ is exactly $x$ itself. 
There may exist several nodes $\nod \in F$ such that
$x \in \Ske(\nod)$, however only one of them is maximal
thanks to \cref{claim:skel:skeletons-totally-ordered}.
As a consequence one can partition
$\Leaves(F)$ depending
on the maximal skeleton (for inclusion) which contains a given leaf
(\cref{def:skel-root}).

\begin{lemma}
	\label{claim:skel:skeletons-totally-ordered}
    Let $F \in \forests{\mu}{d}$,
	$x \in \Leaves(F)$.
    There exists $\nod \in \Nodes(F)$ such that
	$x \in \Ske(\nod)$.

    Furthermore, for every $\nod,\nod'$ such that
	$x \in \Ske(\nod) \cap \Ske(\nod')$, $\Ske(\nod) \subseteq \Ske(\nod')$ or
	$\Ske(\nod') \subseteq \Ske(\nod)$.
\end{lemma}

\begin{definition}
    \label{def:skel-root}
	Let $\intro\partitionskel \colon \Leaves(F) \to \Nodes(F)$
	map a leaf $x$ to the $\nod \in \Nodes(F)$ such that $x \in \Ske(\nod)$
	and $\Ske(\nod)$ is maximal for inclusion.
\end{definition}

Following the work of \cite{doueneau2022hiding},
we define a notion of dependency of leaves
(\cref{def:dependencies})
based on the relationship between
their maximal skeletons
(\cref{def:observation}).

\begin{definition}[Observation]
    \label{def:observation}
	We say that $\nod' \in \Nodes(F)$ \intro{observes} $\nod \in \Nodes(F)$
	if either $\nod'$ is an ancestor of $\nod$
    (this includes $\nod$ itself), or the immediate
	left or right sibling of an ancestor of
	$\nod$.
\end{definition}

\begin{definition}[Dependency]
    \label{def:dependencies}
	In a forest $F$, a leaf $y$ \intro{depends} on a leaf $x$,
    written 
    $x \dependson y$,
	when $\partitionskel(y)$ \kl{observes} $\partitionskel(x)$.
\end{definition}

Beware that the relation $x \dependson y$
is not symmetric. This allows us to ensure that the number of
leaves $y$ that depend on a fixed leaf $x$ is uniformly bounded.
\begin{claim}
	\label{rem:skel:bound-depnodes}
	Given $d \ge 0$, there exists a (computable) bound $N_d \in \Nat$
	such that for all $F \in \forests{\mu}{d}$
	and all leaf $x \in \Leaves(F)$, there exist at most
	$N_d$ leaves which \kl{depend} on $x$.
\end{claim}

It is a routine check that for every fixed $d$,
one can define the predicate
$\msodep(x,y)$ in $\MSO$ over $\forests{\mu}{d}$ checking whether
$x \dependson y$ or $y \dependson x$, that is the \emph{symmetrised}
version of $x \dependson y$.
We generalize this predicate to tuples $\vec{x} \defined (x_1, \dots, x_k)$
via:
\begin{equation*}
	\intro\msodep(\vec{x}) \defined \left\{
	\begin{array}{ll}
		\top                                                  & \mbox{for } k=0;  \\
		\top  \mbox{ if and only if } x_1 \mbox{ is the root} & \mbox{for } k=1 ; \\
		\bigvee_{i \neq j} \msodep(x_i,x_j)                   & \mbox{otherwise.}
	\end{array}
	\right.
\end{equation*}

Notice that the independence (or dependence) of a tuple of leaves
$\vec{x}$ only depends on the tuple $\partitionskel(x_1),
	\dots, \partitionskel(x_n)$.
The notion of dependent leaves is motivated by the fact
that counting dependent leaves can be done with one variable less,
as shown in \cref{lem:skel:inv-decomp}.

\begin{lemma}
	\label{lem:skel:inv-decomp}
	Let $d \ge 0$, $M$ be a finite monoid,
    $\mu \colon A^* \to M$, $k \geq 1$, and $\psi \in \INV_k$.
    One can effectively build a function $g : (\hat{A})^* \fonc \Rel \in \Ve_{k-1}$
    such that for every $F \in \forests{\mu}{d}$,
	$g(F) = \card (\psi(\vec{x}) \wedge \msodep(\vec{x}))(F)$.
\end{lemma}

\begin{definition}
    \label{def:fdep-findep}
    Let $k \ge 1$ and $f \in \Ve_k$,
    thanks to \cref{lem:skel:to-forest}
    and \cref{thm:simon},
    there exists $\mu : A^* \fonc M$, $d \ge 0$, $\delta_i \in \Rel$,
    $\psi_i \in \INV_k$ such that:
    \begin{equation}
        \label{eq:findep}
        \begin{aligned}
            f & = \left(\sum_{i=1}^n \delta_i \card \psi_i(\vec{x}) \right) \circ \repforest               \\
              & = \underbrace{\left(\sum_{i=1}^n \delta_i \card (\psi_i(\vec{x})
              \wedge \msodep(\vec{x})) \right)}_ %
            {\defined \intro\fdep} \circ \repforest
            \\
              &+
            \underbrace{\left(\sum_{i=1}^n \delta_i \card
            (\psi_i(\vec{x}) \wedge \neg\msodep(\vec{x})) \right)}_%
            {\defined \intro\findep} \circ \repforest.
        \end{aligned}
    \end{equation}
    We say that $\fdep$ is the \emph{dependent part} of $f$
    and $\findep$ is its \emph{independent part}.
\end{definition}

Thanks to \cref{lem:skel:inv-decomp,prop:precompose},
for every $k \geq 1$ and $f \in \Ve_k$,
$(\fdep \circ \repforest) \in \Ve_{k-1}$
(over $\forests{\mu}{d}$).
Hence, whether the function $f$ belongs to $\Ve_{k-1}$
only depends on its independent part.
We will actually prove that in this case,
$f \in \Ve_{k-1}$  if and only if $\findep = 0$.
For that, we will rely on ``pumping families''
that respect $\repforest$.

\begin{definition}[Pumping family]
	A $(\mu, d)$-\emph{pumping family} of size $k \geq 1$ is given by
	words $\alpha_0, w_1, \alpha_2, \dots, \alpha_{k-1}, w_k, \alpha_k \in
    A^*$,
    together with a family $F^{\vec{X}}$ of forests
    in $\forests{\mu}{d}$,
    such that for all $1 \leq i \leq k$, $w_i \neq \varepsilon$,
    and
    $F^{\vec{X}}$ is a $\mu$-forest
    of 
	$w^{\vec{X}} \defined \alpha_0 \prod_{i = 1}^{k} (w_i)^{X_i} \alpha_i$
    for every 
	$\vec{X} \defined X_1, \dots, X_k \ge 0$.
\end{definition}

\begin{remark}
	A $(\mu,d)$-\kl{pumping family} of size $k$ satisfies that
	$|w^{\vec{X}}| = \Theta(X_1 + \dots + X_k)$,
	and $|F^{\vec{X}}| = \Theta(X_1 + \dots + X_k)$
    since the depth of $F^{\vec{X}}$ is bounded by $d$.
\end{remark}

\begin{lemma}
	\label{lem:skel:indep-iter-zero}
	Let $\findep$ be defined as in \cref{eq:findep}. Then,
	$\findep \neq 0$ if and only if there exists a $(\mu,d)$-\kl{pumping family} of size $k$
	such that $f(F^{\vec{X}})$ is ultimately 
	a $\Rel$-polynomial in $X_1, \dots, X_k$
	with a non-zero coefficient for $X_1 \cdots X_k$.

	Moreover, one can decide whether $\findep = 0$.
\end{lemma}

Now, we are almost ready to conclude the proof of \cref{lem:skel:semantic-cond}.
The only difficulty left is
handled by the following technical lemma which enables
to lift a bound on the asymptotic growth of polynomials
to a bound on their respective degrees. It is also
reused in \cref{sec:aperiodic}.

\begin{lemma}
	\label{lem:skel:lemmapoly}
	Let $P,Q$ be two polynomials in $\mathbb{R}[X_1, \dots, X_n]$.
    If
	$|P| = \bigO(|Q|)$,
	then $\deg (P) \leq \deg (Q)$.
\end{lemma}

\begin{proof}[Proof of \cref{lem:skel:semantic-cond}]
	The only non-trivial implication is
	\cref{it:cond:tuple} $\implies$ \cref{it:cond:class}.
	Let $f \in \Ve_k$ satisfying the conditions of \cref{it:cond:tuple}.
	We can decompose this function following \cref{eq:findep}.
	As observed above, we only need to show that
	$\findep = 0$.

	Consider a \kl{pumping family}
	$(w^{\vec{X}}, F^{\vec{X}})$
	of size $k$,
	we have:
	\begin{equation*}
		|\findep(F^{\vec{X}})| 
		= |f(w^{\vec{X}}) - \fdep(F^{\vec{X}})|
		= \bigO(|X_1 + \dots + X_k|^{k-1}).
	\end{equation*}
	Assume by contradiction that $\findep \neq 0$,
    \cref{lem:skel:indep-iter-zero}
	provides us with
	a \kl{pumping family} such that
	$\findep(F^{\vec{X}})$ is ultimately a polynomial
	with non-zero coefficient for $X_1 \cdots X_k$. 
	As this polynomial is asymptotically bounded
    by $(X_1 + \dots + X_k)^{k-1}$,
	\cref{lem:skel:lemmapoly} yields a contradiction.

	The constructions of $\repforest$, $\fdep$, and $\findep$ are effective,
	therefore so is our procedure. Moreover, one can decide
	whether $\findep = 0$ thanks to \cref{lem:skel:indep-iter-zero}.
\end{proof}

%% file: parts/residual.tex
\section{Residual Transducers}
\label{sec:residual}

In this section, we provide a canonical object associated
to any $\Rel$-polyregular function, named its \emph{residual transducer}.
Our construction is effective, and the algorithm heavily relies on \cref{thm:skel:pebblemin}.
This new object has its own interest, and it will also
be used in \cref{sec:aperiodic} to decide \emph{first-order definability}
of $\Rel$-polyregular functions, that will extend
first-order definability for regular languages (see e.g. \cite{perrin1986first}
for an introduction).

\subsection{Residuals of a function}

We first introduce the notion of \kl{residual} of a
function $f \colon A^* \to \Rel$ under a word $u \in A^*$.

\begin{definition}[Residual] Given $f \colon A^* \to \Rel$ and $u \in A^*$,
	we define the function $\push{f}{u} \colon A^* \to \Rel, w \mapsto f(uw)$.
    We let $\Res(f) \defined \setof{\push{f}{u}}{u \in A^*}$ be the
    set of \emph{residuals} of $f$.
\end{definition}

\begin{example}
\label{ex:residuals-square}
The residuals of the function $w \mapsto |w|^2$
are the functions $w \mapsto |w|^2 + 2n |w| + n^2$ for $n \ge 0$.
\end{example}

\begin{example}
    \label{ex:residual-2-n}
    The residuals of the function $w \mapsto (-2)^{|w|}$
    are 
    the functions 
    $w \mapsto (-2)^{n + |w|}$ for $n \geq 0$.
\end{example}

It is easy to see that $u \mapsto \push{f}{u}$
defines a monoid action of $A^*$ over $A^* \fonc \Rel$.
Let us observe that this action
(effectively) preserves the classes of functions $\Ve_k$.

\begin{claim} \label{claim:push-pk} Let $k \ge 0$, $f \in \Ve_k$
	and $u \in A^*$. Then $\push{f}{u} \in \Ve_k$
	and this result is effective.
\end{claim}

\begin{remark}[{\cite[Corollary 5.4 p 14]{berstel2011noncommutative}}]
    \label{remark:finite-dimension}
    Let $f \colon A^* \to \Rel$, this function
    is a $\Rel$-rational series
    if and only if
	$\Vect{\Rel}(\Res(f))$
    has finite dimension.
\end{remark}

Note that if $L \subseteq A^*$ and $u \in A^*$, then $\push{\car{L}}{u}$
is the characteristic function of the well-known
residual language $u^{-1}L \defined \setof{w \in A^*}{uw \in L}$.
In particular, the set $\setof{\push{\car{L}}{u}}{u \in A^*}$ is finite
if and only if $L$ is regular.
However, given  $f \in \Ve_k$ for $k \ge 1$, the set $\setof{\push{f}{u}}{u \in A^*}$
is not finite in general (see e.g. \cref{ex:residuals-square}).
We now intend to show that this set is still finite, up to an
identification of the functions whose difference is in  $\Ve_{k-1}$.

\begin{definition}[Growth equivalence] \label{def:equivalence}
	Given $k \ge -1$ and $f,g  \colon A^* \fonc \Rel$,
	we let $f \polysim{k} g$ if and only if $f-g \in \Ve_k$
\end{definition}

Let us observe that $\polysim{k}$  is an equivalence relation,
that is compatible with external $\Rel$-products, sums, $\cauchy$ and $\push{}{}$.

\begin{claim} \label{claim:properties-resi}
	For all $k \ge {-}1$,  $\polysim{k}$ is an equivalence relation
	and the following holds for all $u \in A^*$, $\delta \in \Rel$,
    and $f,g \colon A^* \fonc \Rel$:
	\begin{itemize}
	\item if $f \polysim{k} g$, then $\push{f}{u} \polysim{k} \push{g}{u}$;
	\item $\push{(\car{L} \cauchy f)}{u}
		\polysim{k} (\push{\car{L}}{u}) \cauchy f$ for $L \subseteq A^*$;
    \item if $f \polysim{k} g$ and $f' \polysim{k} g'$ then
        $f + f' \polysim{k} g + g'$;
    \item if $f \polysim{k} g$ then $\delta \cdot f \polysim{k} \delta \cdot g$.
	\end{itemize}
\end{claim}


By combining these results with
the characterization of $\Ve$ 
via these combinators in \cref{theo:rational-polygrowth},
we can show that a function $f \in \Ve_k$ has a finite number of residuals,
up to $ \polysim{k-1}$ identification.

\begin{lemma}[Finite residuals]
	\label{lem:resifini}
	Let $k \geq 0$ and $f \in \Ve_{k}$,
	then the quotient set $\Res(f) / \polysim{k-1}$ is finite.
\end{lemma}

\begin{remark}
    \Cref{ex:residual-2-n}
    exhibits
    a $\Rel$-rational series
    $f$ such that $\Res(f) / \polysim{k}$
    is infinite for all $k \geq 0$.
\end{remark}

Finally, we note that $\polysim{k}$ is decidable in $\Ve$. 

\begin{claim}[Decidability] Given $k \ge {-}1$
and $f,g \in \Ve$, one can decide whether
$f \polysim{k} g$ holds.
\end{claim}

\begin{proof} Let $f,g \in \Ve$.
For $k \ge 0$, $f \polysim{k} g$ if and only if $|(f-g)(w)| = \bigO(|w|^k)$
and this property is decidable by \cref{thm:skel:pebblemin}.
For $k=-1$, we have $f \polysim{k} g$ if and only if $f=g$,
which is decidable by \cref{cor:equivalence}.
\end{proof}


\subsection{Residual transducers}

Now we intend to show that a function $f \in \Ve_k$ can
effectively be computed by a canonical machine,
whose states are based on the finite set $\Res(f)/\polysim{k-1}$,
in the spirit of the residual automaton of a regular language.
First, let us introduce an abstract notion of transducer
which can call functions on suffixes of its input
(this definition is inspired by the \emph{marble transducers} of
\cite{doueneau2020register}, that call functions on prefixes).

\begin{definition}[$\Oras$-transducer]
	Let $k \ge 0$ and $\Oras $ be a fixed subset of the functions
	$ A^* \fonc \Rel$. A \emph{$\Oras$-transducer}
	$\trans = (A, Q,q_0, \delta,\Oras, \lambda,F)$ consists of:
	\item
	\begin{itemize}
		\item a finite input alphabet $A$;
		\item a finite set of states $Q$ with $q_0 \in Q$ initial;
		\item a transition function $\delta \colon Q \times A \fonc Q$;
		\item a labelling function $\lambda \colon Q \times A \fonc \Oras$;
		\item an output function $F \colon Q \fonc \Rel$.
	\end{itemize}
\end{definition}

Given $q \in Q$, we define by induction on $w \in A^*$ the value
$\trans_q(w) \in \Rel$. For $w = \movi$, we let $\trans_q(w) \defined F(q)$.
Otherwise let $\trans_q(aw) \defined \trans_{\delta(q,a)}(w) + \lambda(q,a)(w)$.
Finally, the function computed by the $\Oras$-transducer $\trans$ is defined
as $\trans_{q_0} \colon A^* \fonc \Rel$.
Observe that all the functions $\trans_q$ are total.

Let us recall the standard definition of $\delta^*$ via
$\delta^*(q, ua) \defined \delta(\delta^*(q,u),a)$
and $\delta^*(q,\movi) = q$. Using this notation,
a simple induction shows that
$\trans_{q}(w) = \sum_{uav = w} \lambda(\delta^*(q,u), a) (v) + F(\delta^*(q,w))$.
As a consequence, $\Oras $-transducers
are closely related to Cauchy products.




\begin{example}
	\label{ex:res:1aA-bad}
	We have depicted in \cref{fig:bad-aA}
	a $\Ve_{-1}$-transducer and a $\Ve_{0}$-transducer computing the
	function $\car{a A^*}$ for $A = \{a,b\}$. The first one can easily be
	identified with the minimal automaton of $\car{a A^*}$ (up to considering
	that a state is final if it outputs $1$). The second one has a single state and
	it ``hides'' its computation into the calls to $\Ve_{0}$.
	One can check e.g. that $1 = \car{a A^*}(aab) = (1 - \car{a A^*}(ab)) + (1 - \car{a A^*}(b)) - \car{a A^*}(\movi) + 0$.
\end{example}

\begin{figure}[h!]
	\centering
	\begin{subfigure}[b]{0.45\textwidth}
		\centering
		\begin{tikzpicture}{scale=1}
			\node [state, initial, initial text=, inner sep=3pt, minimum size=0pt,
				accepting above, accepting text = $0$ ]
			(q0) at (0,0) {\small $q_0$};
			\node[state ,  inner sep=3pt, minimum size=0pt, accepting right,
				accepting text = $1$ ] at (2,0.5) (q1) {$q_1$};
			\node[state, initial text=, inner sep=3pt, minimum size=0pt,
				 accepting right, accepting text = $0$ ]%
			at (2,-0.5) (q2) {$q_2$};
			\draw[->] (q0) edge[bend left = 30] node[above]{$a~|~0$} (q1);
			\draw[->] (q0) edge[bend left = -30] node[below]{$b~|~0$} (q2);
			\draw[->]      (q2) edge[loop below] node{${a,b~|~0}$} (q2);
			\draw[->]      (q1) edge[loop above] node{${a,b~|~0}$} (q1);
		\end{tikzpicture}
		\caption{A $\Ve_{-1}$-transducer computing $\car{aA^*}$.}
	\end{subfigure}
	\begin{subfigure}[b]{0.45\textwidth}
		\centering
		\begin{tikzpicture}{scale=1}
			\node [state, initial, initial text=, inner sep=3pt, minimum size=0pt,
				accepting right, accepting text = $0$ ]
			(q0) at (0,0) {\small $q_0$};
			\draw[->]      (q0) edge[loop above] node{${a~|~1-\car{aA^*}}$} (q0);
			\draw[->]      (q0) edge[loop below] node{${b~|~-\car{aA^*}}$} (q0);
		\end{tikzpicture}
		\caption{A $\Ve_0$-transducer computing $\car{aA^*}$}
	\end{subfigure}
	\caption{\label{fig:bad-aA} Two transducers computing $\car{aA^*}$.}
\end{figure}%

The reader may guess that every function  $f \in \Ve_k$ can
effectively be computed by a  $\Ve_{k-1}$-transducer.
We provide a stronger result and show that $f$ can be computed by
some specific  $\Ve_{k-1}$-transducer whose transition function is
uniquely defined by $\Res(f)/\polysim{k-1}$.

\begin{definition} Let $k \ge 0$,
	let $\trans = (A, Q,q_0, \delta, \Oras, \lambda, F)$
	be a $\Ve_{k-1}$-transducer and $f \colon A^* \to \Rel$.
	We say that $\trans$ is a \emph{$k$-residual transducer}
	of $f$ if the following conditions hold:
	\begin{itemize}
		\item $\trans$ computes $f$;
		\item $Q = \Res(f)/\polysim{k-1}$;
        \item for all $w \in A^*$, 
            $\push{f}{w} \in \delta^*(q_0,w)$;
		\item $\lambda(Q,A) \subseteq \Vect{\Rel}(\Res(f)) \cap \Ve_{k-1}$.
	\end{itemize}
\end{definition}

Given a regular language $L$, the $0$-residual transducer of its
indicator function $\mathbf{1}_{L}$ can easily be identified with the
\emph{minimal automaton} of the language $L$, like in \cref{ex:res:1aA-bad}. However,
for $k \ge 1$, the \kl{$k$-residual transducer} of $f \in \Ve_k$ may not be unique.
More precisely, two $k$-residual transducers share the same
underlying automaton $(A,Q, \delta, \lambda)$, but the labels
$\lambda$ of the transitions may not be the same.

\begin{example} The $\Ve_{-1}$-transducer (resp. $\Ve_0$-transducer)
	from \cref{fig:bad-aA} is a $0$-residual transducer
	(resp. $1$-residual transducer) of $\car{aA^*}$.
	Let us check it for the $1$-residual transducer.
	First note that $\push{\car{aA^*}}{b} \polysim{0} \push{\car{aA^*}}{a} \polysim{0} \car{aA^*}$,
	hence $|\Res(\car{aA^*}) / \polysim{0}|=1$. Thus a $1$-residual transducer
	of $\car{aA^*}$ has exactly one state $q_0$.
	Furthermore the labels of the transitions
	of our transducer belong to $\lambda(Q,A) \subseteq \Vect{\Rel}(\Res_f(a))$
	since $1 - \car{aA^*} = (\push{\car{aA^*}}{a}) - \car{aA^*}$.
\end{example}

\begin{example} Let $A \defined \{a,b\}$.
	The function $f \colon w \mapsto |w|_a \times |w|_b \in \Ve_2$
	has a single residual up to $\polysim{1}$-equivalence.
	A $2$-residual transducer of $f$ is depicted
	in \cref{fig:2-res:ab}.
\end{example}

\begin{example}
	Let $A \defined \{a\}$. The function $g \colon w \mapsto (-1)^{|w|} \times |w| \in \Ve_1$
	has two residuals up to $\polysim{0}$-equivalence.
	A $1$-residual transducer of $g$ is depicted
	in \cref{fig:2-res:-1}.
\end{example}

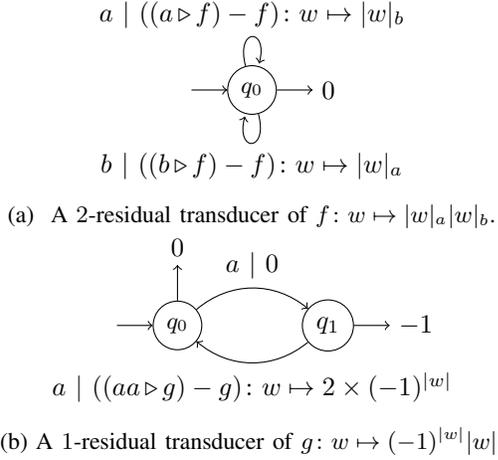
\begin{figure}[h!]
	\centering
	\begin{subfigure}[b]{0.45\textwidth}
		\centering
		\begin{tikzpicture}{scale=1}
			\node [state, initial, initial text=, inner sep=3pt, minimum size=0pt,
				accepting right, accepting text = $0$ ]
			(q0) at (0,0) {\small $q_0$};
			\draw[->]      (q0) edge[loop above] node{${a~|~((\push{f}{a}) -f)
                \colon w \mapsto |w|_b}$} (q0);
			\draw[->]      (q0) edge[loop below] node{${b~|~((\push{f}{b}) -f)
                \colon w \mapsto |w|_a}$} (q0);
		\end{tikzpicture}
		\caption{\label{fig:2-res:ab} A $2$-residual transducer of $f \colon w \mapsto |w|_a |w|_b$.}
	\end{subfigure}
	\begin{subfigure}[b]{0.45\textwidth}
		\centering
		\begin{tikzpicture}{scale=1}
			\node [state, initial, initial text=, inner sep=3pt, minimum size=0pt,
				accepting above, accepting text = $0$ ]
			(q0) at (0,0) {\small $q_0$};
			\node[state ,  inner sep=3pt, minimum size=0pt, accepting right,
				accepting text = $-1$ ] at (2,0) (q1) {$q_1$};
			\draw[->] (q0) edge[bend left = 40] node[above]{$a~|~0$} (q1);
			\draw[->] (q1) edge[bend left = 40]
                node[below]{$a~|~((\push{g}{aa})-g) \colon
						w \mapsto 2 \times (-1)^{|w|}$} (q0);
		\end{tikzpicture}
		\caption{\label{fig:2-res:-1}A $1$-residual transducer of $g \colon w \mapsto (-1)^{|w|} |w|$.}
	\end{subfigure}
	\caption{\label{fig:residuals} Two residual transducers.}
\end{figure}%

\begin{figure}[h!]
    \begin{center}
        \begin{tikzpicture}[
            etat/.style={
                state,
                inner sep=3pt,
                minimum size=0pt,
                font={\small},
                on grid,
            },
            begin/.style={
                initial,
                initial text=,
            },
            created/.style={
            },
            phantom/.style={
                dashed, thick, red
            },
            equivalent/.style={
                dotted, thick, blue
            },
            etiquette/.style={
                state,
                font={\tiny},
                inner sep=1pt,
                minimum size=0pt,
                fill=white
            }
            ]

            \node[etat, created, begin,
                accepting above, accepting text=$f(\movi)$
                ] (f) at (0,0) {$\push{f}{\movi}$};
            \node[etat, created, above right=2.5cm of f,
                accepting above, accepting text=$f(a)$
                ] (fa) {$\push{f}{a}$};
            \node[etat, created, below right=2.5cm of f,
                accepting right, accepting text=$f(b)$
                ] (fb) {$\push{f}{b}$};
            \node[etat, phantom, right=2cm of fa,
                accepting above, accepting text=$f(aa)$
                ] (faa)
                {$\push{f}{aa}$};

            \draw[created, ->] (f) -- node[midway, below right] {\small $a \mid 0$} (fa);
            \draw[created, ->] (f) -- node[midway, below left] {\small $b \mid 0$} (fb);
            \draw[phantom, ->] (fa) -- node[midway, above] {\small $a \mid 0$}
            (faa);
            \draw[equivalent, ->] (fa) -- node[midway, right] {\small $a \mid \push{f}{aa} -
                \push{f}{b}$} (fb);

            \node[etiquette, above right=-3pt of f] {$0$};
            \node[etiquette, above right=-3pt of fa] {$1$};
            \node[etiquette, above right=-3pt of fb] {$2$};
            \node[etiquette, above right=-3pt of faa] {$3$};

            \node[phantom,
                  right = 2cm of f, yshift=-1cm]
                  {$f(aa) = [f(aa) - f(b)] + f(b)$};

          \begin{scope}[on background layer]
                \clip (-1,3.5) rectangle (3,-3);
                \draw[draw=none,fill=blue!10!white] (-3,0) circle (6);
                \draw[draw=none,fill=green!10!white] (-3,0) circle (4);
                \node (Q) at (-0.5,-2.5) {$Q$};
                \node (O) at (2,-2.5) {$O$};
          \end{scope}
        \end{tikzpicture}
    \end{center}
	\caption{
        \label{fig:computing-residual-automata}
        Example of a partial execution of \cref{algo:residual} to build
        a $k$-residual transducer of a
        function $f \colon A^* \to \Rel$
        such that
        $\push{f}{aa} \polysim{k} \push{f}{b}$.
        Nodes are labelled by their creation time.
        At this stage, $Q = \set{ \push{f}{\movi} }$,
        $O = \set{ \push{f}{a}, \push{f}{b} }$.
        The red node is not created, and the blue transition
        is added instead, corresponding to the ``else'' branch
        line $10$ of \cref{algo:residual}.
    }
\end{figure}%

Now, let us describe how to build a $k$-residual transducer
for any $f \in \Ve_k$. As an illustration of how
\cref{algo:residual} works, we refer the reader to
\cref{fig:computing-residual-automata}.

	\begin{algorithm}

		$O \defined \{ \push{f}{\movi} \}$;

		$Q \defined \vide$;

		\While{$O \neq \vide$}{

			choose $\push{f}{w} \in O$;

			\For{$a \in A$}{

				\eIf{$\push{f}{wa} \not \polysim{k-1} \push{f}{v} $ \normalfont{for all} $\push{f}{v}\in O \uplus Q$}{

					$O \defined O \uplus \{\push{f}{wa}\}$;

					$\delta(\push{f}{w}, a) \defined \push{f}{wa}$;

					$\lambda(\push{f}{w}, a) \defined 0$;

				}{

					let $\push{v}{f} \in O \uplus Q$ be such that $\push{f}{wa} \polysim{k-1} \push{f}{v} $;

					$\delta(\push{f}{w},a) \defined \push{f}{v}$;

					$\lambda(\push{f}{w}, a) \defined \push{f}{wa} - \push{f}{v}$;

				}

			}

			$O \defined O \smallsetminus \{\push{f}{w}\}$;

			$Q \defined Q \uplus \{\push{f}{w}\}$;

			$F(\push{f}{w}) \defined f(w)$;
			
		}

		\caption{\label{algo:residual} Computing a $k$-residual transducer of $f \in \Ve_k$}
	\end{algorithm}

\begin{lemma} \label{lem:resitrans} Let $k \ge 0$.
	Given $f\colon A^* \to \Rel$ such that $\Res(f) / \polysim{k-1}$ is finite,
	\cref{algo:residual} builds  a \mbox{$k$-residual} transducer of $f$.
	Its steps are effective given $f \in \Ve_k$.
\end{lemma}

\begin{remark} In \cref{algo:residual}, we need to ``choose'' a way to
	range over the elements of $O$ and the letters of $A$.
	Different choices may not lead to the same $k$-residual transducers.
\end{remark}

We deduce from \cref{lem:resitrans} that $\Ve_{k-1}$-transducers
describe exactly the class $\Ve_k$ (\cref{cor:Vk-trans}).

\begin{corollary} \label{cor:Vk-trans} For all $k \ge 0$,
	$\Ve_k$ is the class of functions which
	can be computed by a $\Ve_{k-1}$-transducer.
	Furthermore, the conversions are effective.
\end{corollary}

\begin{corollary}[{To be compared to \cref{remark:finite-dimension}}]
    \label{cor:faux}
	For all $k \ge 0$, $\Ve_{k} = \setof{f \colon A^* \fonc \Rel}{\Res(f) / \polysim{k-1} \text{ is finite}}$.
\end{corollary}

%% file: parts/aperiodic.tex
\section{Star-free $\Rel$-polyregular functions}
\label{sec:aperiodic}

In this section, we study the subclass of
$\Rel$-polyregular functions that are built
by using only $\FO$-formulas, that we call
\emph{star-free $\Rel$-polyregular functions}.
The term ``star-free'' will be justified
in \cref{theo:Zsf-expressions}.
As observed in introduction, very little is known
on deciding $\FO$ definability of functions
(contrary to languages). The main result of this
section shows that we can decide if a $\Rel$-polyregular function
is star-free. Our proof crucially relies on the canonicity
of the residual transducer introduced in \cref{sec:residual}.
We also provide several characterizations
of star-free $\Rel$-polyregular functions,
that specialize the results of \cref{sec:prelim}.

\begin{definition}[Star-free $\Rel$-polyregular]
    \label{def:Zsf}
	For $k \ge 0$, we let
	$\We_k \defined \Vect{\Rel}{{}}(\{ \card \phi : \phi \in \FO_\ell, \ell \leq
    k\})$.
	Let $\intro\We \defined \bigcup_k \We_k$, it is the
	class of \emph{star-free $\Rel$-polyregular functions}.
\end{definition}
We also let $\We_{-1} \defined \{0\}$.
Similarly to $\Ve_k$,
$\We_k = \Vect{\Rel}(\setof{\card\phi}{\phi \in \MSO_k} \cup \{
\car{\set{\movi}} \})$.

\begin{example}[]
$\We_0$ is exactly the set of functions
of the form $\sum_i \delta_i \mathbf{1}_{L_i}$ where the
$\delta_i \in \Rel$ and the $\mathbf{1}_{L_i}$ are indicator
functions of star-free languages
(compare with \cref{ex:Ve0}).
\end{example}

\begin{example}  The function $w \mapsto |w|_a \times |w|_b$
	is in $\We_1$. Indeed, the formulas given in \cref{ex:a-b}
	are in $\FO$.
\end{example}

Now, we give an analogue of \cref{theo:rational-polygrowth}
that characterizes $\We$ as $\Rel$-rational expressions based on
indicators of star-free languages, forbidding the use of the Kleene star.

\begin{theorem}
    \label{theo:Zsf-expressions}
	Let $f: A^* \fonc \Rel$, the following are (effectively) equivalent:
    \begin{enumerate}
		\item \label{it:sf-expr:star} $f$ is a star-free
		$\Rel$-polyregular function;
        \item \label{it:sf-expr:smallest} $f$ belongs to the smallest
		class of functions that contains the indicator functions of all star-free languages
		and is closed under taking external $\Rel$-products, sums and
		Cauchy products.
    \end{enumerate}
\end{theorem}
\begin{proof}
    We apologize for the inconvenience of looking back at
    \cref{lem:vks:inductcauchy} and noticing that
    the property holds mutatis mutandis for first-order formulas.
    In particular, one obtains
    the equivalent of \cref{eq:cauchys}
    of \cref{theo:rational-polygrowth}
    \begin{equation}
        \label{eq:cauchys-starfree}
        \begin{array}{rcll}
             & \We_k & = \Vect{\Rel}(\{\car{L_0} \cauchy \cdots \cauchy
            \car{L_k}& \\ &&\colon L_0, \dots, L_k \text{ star-free languages}\})
        \end{array}
    \end{equation}
    and the result follows.
\end{proof}

\begin{example}
    The function
    $\car{A^*a} \cauchy \car{A^*} \colon w \mapsto |w|_a$ belongs to $\We_1$,
    and the function 
	 $\car{A^* a} \cauchy \car{A^*} \cauchy   \car{bA^*}  +
		\car{A^* b} \cauchy  \car{A^*} \cauchy \car{a A^*} \colon w \mapsto
        |w|_a \times |w|_b$
    belongs to $\We_2$.
\end{example}

\subsection{Deciding star-freeness}
\label{ssec:car-sf}

Now, we intend to show that given a \kl{$\Rel$-polyregular function},
we can decide if it is star-free. Furthermore, we provide
a semantic characterization of \kl{star-free $\Rel$-polyregular functions}
leveraging ultimate $N$-polynomiality.
We recall
(see \cref{def:ultimately-polynomial-general}) that 
a function $f \colon A^* \fonc \Rel$ is 
ultimately $1$-polynomial
when,
for all  $\alpha_0, w_1, \alpha_1, \dots, w_{\ell}, \alpha_{\ell} \in A^*$,
there exists $P \in \Rat[X_1, \dots, X_\ell]$,
such that
$f(\alpha_0 w_1^{X_1} \alpha_1 \cdots w_{\ell}^{X_{\ell}}\alpha_{\ell})
= P(X_1, \dots, X_\ell)$,
for $X_1, \dots, X_\ell$ large enough.
Being ultimately $1$-polynomial generalizes
star-freeness for regular languages, as easily 
observed in \cref{ex:1poly-aperiodic}.

\begin{claim}
\label{ex:1poly-aperiodic}
A regular language $L$ is star-free if and only if $\car{L}$ is
ultimately $1$-polynomial.
\end{claim}

\begin{example}
    \label{ex:non-1-upoly}
    It is easy to see that $w \mapsto |w|_a \times |w|_b$ is
    ultimately $1$-polynomial.
	As a counterexample, recall the map
	$f \colon w \mapsto (-1)^{|w|} \times |w|$.
    The map $f$ is ultimately $2$-polynomial
    because $X \mapsto (-1)^{2X + 1} (2 X + 1)$
    and $X \mapsto (-1)^{2X} 2X$ are both
    polynomials. However, $f$
	is not ultimately $1$-polynomial 
	since $X \mapsto (-1)^X X$ is not a polynomial.
\end{example}

Now, let us state the main theorem of this section.

\begin{restatable}{theorem}{theofodecidable} \label{theo:FO-decidable}
	Let $k \ge 0$ and $f \in \Ve_k$. The following properties are
	(effectively) equivalent:
	\begin{enumerate}
		\item $f \in \We$;
		\item $f \in \We_k$;
		\item $f$ is $1$-\kl{ultimately polynomial}.
	\end{enumerate}
	Furthermore, this property is decidable.
\end{restatable}

Let us observe that  \cref{theo:FO-decidable}
implies an analogue of
\cref{thm:skel:pebblemin} for the classes $\We_k$.
We conjecture that a direct proof of \cref{cor:FO-min}
is possible. However, such a proof cannot rely on
factorizations forests (that cannot be built in $\FO$),
and it would require a  (weakened) notion of $\FO$-definable
factorization forest as that proposed in \cite{colcombet2022countable}.

\begin{corollary}
    \label{cor:SF-POLY-k}
    $\We_k = \We \cap \Ve_k$.
\end{corollary}

\begin{corollary}[$\FO$ free variable minimization]
	\label{cor:FO-min}
	Let $f \in \We$, then $f \in \We_k$ if and only if
	\kl{$|f(w)| = \bigO(|w|^k)$}.
	This property is decidable and the construction
	is effective.
\end{corollary}

\begin{proof} Let $f \in \We$ be such that
	$|f(w)| = \bigO(|w|^k)$.
	By \cref{thm:skel:pebblemin}
	we get $f \in \Ve_k$, thus by \cref{theo:FO-decidable},
	$f \in \We_k$. All the steps are effective and decidable.
\end{proof}

The rest of \cref{ssec:car-sf} is devoted to sketching the proof
of \cref{theo:FO-decidable}. Given $f \in \Ve_k$,
the main idea is to use its \kl{$k$-residual transducer}
to decide whether $f \in \We_k$. Indeed, this transducer
somehow contains intrinsic information
on the semantic of $f$. We show that  \kl{star-freeness}
faithfully translates to a \kl{counter-free} property of the
$k$-residual transducer, together with an inductive property
on the labels of its transitions.

\begin{definition}[Counter-free]
	A deterministic automaton $(A,Q,q_0, \delta)$ is \emph{counter-free}
	if for all $q \in Q$, $u \in A^*$, $n \ge 1$, if $\delta(q,u^{n})  = q$ then $\delta(q,u)  = q$
	 (see e.g. \cite{mcnaughton1971counter}).
     We say that a $\Oras$-transducer is \emph{counter-free} if
	its underlying automaton is so.
\end{definition}

\begin{example} The $\Ve_0$-transducer depicted in \cref{fig:2-res:-1}
	is not counter-free, since  $\delta(q_0, aa) = q_0$
	but $\delta(q_0,a) \neq q_0$.
\end{example}

\Cref{theo:FO-decidable} is a direct consequence of the more precise
\cref{theo:big-induction}. Note that the semantic characterization (\cref{it:big:up})
is not a side result: it is needed within the inductive proof
of equivalence between the other items.

\begin{theorem} \label{theo:big-induction}
	Let $k \ge 0$ and $f \in \Ve_k$, the following conditions are equivalent:
	\begin{enumerate}
		\item \label{it:big:SF} $f \in \We$;
		\item \label{it:big:up} $f$ is ultimately $1$-polynomial;
		\item \label{it:big:kres} for all \kl{$k$-residual transducer} of $f$,
			this transducer is \kl{counter-free} and has labels in $\We_{k-1}$;
		\item \label{it:big:kres2} 
		      there exists a counter-free $\We_{k-1}$-transducer
		      that computes $f$;
		\item \label{it:big:sf-k} $f \in \We_k$.
	\end{enumerate}
	Furthermore, this property is decidable and
	the constructions are effective.
\end{theorem}

The proof of  \cref{theo:big-induction} will be done
by induction on $k \ge 0$. First, 
let us note that a counter-free transducer
computes a star-free function (provided that
the labels are star-free).

\begin{lemma}
	\label{lem:aper:counterfreefo}
	Let $k \ge 0$, a \kl{counter-free} \kl{$\We_{k-1}$-transducer}
	(effectively) computes a function of $\We_k$.
\end{lemma}

We show that \kl{star-freeness} implies ultimate
$1$-polynomiality. This result  generalizes
ultimately $1$-polynomiality of
the characteristic functions of star-free languages
 (see \cref{ex:1poly-aperiodic}).

%
%
%
%


\begin{lemma} \label{lem:sf-up} Let $f \in \We$, then $f$
	is ultimately $1$-polynomial.
\end{lemma}
\begin{proof}
    From \cref{ex:1poly-aperiodic} we get that $\car{L}$ is ultimately $1$-polynomial
    if $L$ is star-free. The result therefore immediately follows from \cref{theo:Zsf-expressions}
    and \cref{lem:Zpoly-pump-first}.
\end{proof}

Last but not least, we show that ultimate $1$-polynomiality
implies that any $k$-residual transducer is counter-free.
\Cref{lem:up-cf} is the key ingredient for showing \cref{theo:big-induction}.

\begin{lemma} \label{lem:up-cf} Let $k \ge 0$. 
    Let $f \in \Ve_k$ which is ultimately $1$-polynomial
	and $\trans$ be a $k$-residual transducer of $f$.
	Then $\trans$ is counter-free and its label functions are
    ultimately $1$-polynomial.
\end{lemma}

\begin{proof}[Proof of \cref{theo:big-induction}]
	The (effective) equivalences are shown by
	induction on $k \ge 0$. 
    For \cref{it:big:sf-k} $\implies$ \cref{it:big:SF}, the
    implication is obvious.
    For \cref{it:big:SF} $\implies$ \cref{it:big:up}
	we apply \cref{lem:sf-up}. For \cref{it:big:up} $\implies$
	\cref{it:big:kres}, we use  \cref{lem:up-cf} which shows that
	any \kl{$k$-residual transducer} of $f$ is \kl{counter-free}
	and has ultimately $1$-polynomial labels. Since these labels are
	in $\Ve_{k-1}$, then by induction hypothesis they belong
	to $\We_{k-1}$. 
    For \cref{it:big:kres} $\implies$ \cref{it:big:kres2},
    the result follows because there exists a $k$-residual transducer
    computing $f$.
    For  \cref{it:big:kres2} $\implies$
	\cref{it:big:sf-k} we use \cref{lem:aper:counterfreefo}.

	It remains to see that this property can be decided,
	which is also shown by induction on $k \ge 0$.
	Given $f \in \Ve_k$, we can effectively build a \kl{$k$-residual transducer}
	of $f$ by \cref{lem:resitrans}. If it is not \kl{counter-free}, the function is
	not \kl{star-free polyregular}. Otherwise, we can check by induction
	that the labels belong to $\We_{k-1}$ (since they belong to $\Ve_{k-1}$).
\end{proof}

\subsection{Relationship with polyregular functions
and rational series}

Let us now specialize the multiple characterizations of
$\Ve$ presented
in \cref{sec:prelim}
to $\We$, which completes the third column of
\cref{tab:summary-characterisations}.


Bojańczyk \cite[page 13]{bojanczyk2018polyregular}
 introduced the notion of \emph{first-order (definable) polyregular functions}.
It is an easy check that star-free $\Rel$-polyregular functions
are obtained by post composition with $\polysum$, in a similar way
as \cref{prop:mikolaj1bis}.

\begin{proposition}
	\label{prop:mikolaj2}
	The class $\We$ is (effectively) the class
	of functions $\polysum \circ f$ where $f : A^* \fonc \{\pm1\}^*$
	is first-order polyregular.
\end{proposition}

Now, let us provide a description
of $\We$ in terms of eigenvalues
in the spirit of \cref{theo:Zpoly-Zrat}.
Intuitively, it shows that
a linear representation $(I, \mu,F)$
computes a function in $\We$
if and only if $\Spec(\mu(A^*))$ contains no
non-trivial subgroup,
mimicking the notion of \emph{aperiodicity} for monoids
\footnote{Beware: the spectrum of a linear representation
    may not be a semigroup.}.

\begin{theorem}[Star-free]
	\label{theo:Zsf-eigen}
	Let $f: A^* \fonc \Rel$, the following are (effectively) equivalent:
	\item
	\begin{enumerate}
		\item \label{it:sf:star} $f$ is a star-free
		$\Rel$-polyregular function;
		\item \label{it:sf:forall} $f$ is a $\Rel$-rational series
		and for all minimal linear representation $(I, \mu,F)$ of $f$,
		$\Spec(\mu(A^*)) \subseteq \{0,1\}$;
		\item \label{it:sf:exists} $f$ is a $\Rel$-rational series
		and there exists a linear representation $(I, \mu,F)$ of $f$
		such that $\Spec(\mu(A^*)) \subseteq \{0,1\}$.
	\end{enumerate}
\end{theorem}

\begin{proof}
    For \cref{it:sf:forall} $\implies$ \cref{it:sf:exists}, the result
    is obvious.
    
    For \cref{it:sf:star} $\implies$ \cref{it:sf:forall},
	consider a minimal presentation of $f$
	using $(I, \mu, F)$ of dimension $n$.
	Then consider a word $w$, $\lambda$ a complex eigenvalue
	of $\mu(w)$.
    Thanks to \cref{lem:capturing-eigenvalues},
    there exists $w, \alpha_{i,j}, u_i, v_j \in A^*$ for $1 \leq i,j \leq n$
    such that
	$
    \lambda^X = \sum_{i,j = 1}^n \alpha_{i,j} f(v_i w^X u_j)
    $.
    Because $f \in \We$, $f$ is ultimately $1$-polynomial thanks
    to \cref{theo:big-induction}.
	This entails that $X \mapsto \lambda^X$
	is a polynomial for $X$ large enough.
	Therefore, $\lambda \in \{0,1\}$.

    For \cref{it:sf:exists} $\implies$ \cref{it:sf:star},
	let us prove that the computed function is
    ultimately $1$-polynomial, which is enough thanks
    to \cref{theo:big-induction}.
    Because the eigenvalues of the matrix $\mu(w) \in \Mat{n,n}(\Rel)$
	for $w \in A^*$ are all in $\{0,1\}$,
	its characteristic polynomial splits over $\Rat$, hence
    there exists $P \in \Mat{n,n}(\mathbb{Q})$ such that 
    $T \defined P M_w P^{-1}$ is
	upper triangular with diagonal values in $\{0,1\}$.
	In particular, $\mu(w)^X = P^{-1} T^X P$,
	but a simple induction proves that the coefficients of $T^X$
	are in $\Rat[X]$  for large enough $X$, hence so does $\mu(w)^X$.
	Pumping multiple patterns at once only computes sums of products
	of polynomials, hence the function is ultimately
	$1$-polynomial. Thanks to \cref{theo:big-induction}, it is
	star-free $\Rel$-polyregular.
\end{proof}

\begin{remark}
    \label{remark:1-polynomial-monomial}
When showing
 \cref{it:sf:exists} $\implies$ \cref{it:sf:star},
 we have in fact shown that the following weaker form
 of ultimate $1$-polynomiality characterizes 
 $\We$ among $\Rel$-rational series:
for all  $u,w,v \in A^*$,
there exists $P \in \Rat[X]$,
such that $f(u w^{X} v) = P(X)$,
for $X$ large enough. 
\end{remark}

Beware that 
\cref{remark:1-polynomial-monomial}
slightly differs from
\cref{remark:N-polynomial-monomial}:
the latter deals with a polynomial upper bound, whereas
an equality is needed to
characterize star-freeness.

\begin{example}
    Let $u,v,w \in A^*$,
    then $|\Lodd(u w^{X} v)| \leq 1$ for every $X \geq 0$.
    However, $\Lodd \not\in \We$.
\end{example}

As a concluding example, let us observe that our notion of
star-free $\Rel$-polyregular  functions
differs from the functions definable in
the \emph{weighted first order logic}
introduced by Droste and Gastin
\cite[Section 4]{droste2019aperiodic}
when studying rational series.

\begin{example}
	\label{ex:gastin}
    Thanks to \cite[Theorem 1]{droste2019aperiodic}, the map
    $f \colon w \mapsto (-1)^{|w|} |w|$ is definable in weighted first order logic
    (however, $f \not\in \We$ as shown in \cref{ex:non-1-upoly}).
    Similarly, the indicator function $\Lodd$ is also
    definable in weighted first order logic, even though
    the language of words of odd length is not star-free.
\end{example}

%% file: parts/conclusion.tex
\section{Outlook}

This paper describes a robust class of functions,
which admits several characterizations in terms
of logics, rational expressions, rational series and transducers.
Furthermore, two natural class membership problems
(free variable minimization and first-order definability) are shown decidable.
We believe that these results together with the technical tools
introduced to prove them open the range towards
a vast study  of $\Rel$- and $\Nat$-polyregular functions.
Now, let us discuss a few tracks which
seem to be promising for future work.

\subparagraph*{Weaker logics} Boolean combinations of existential
first-order formulas define a well-known subclass of first-order logic,
often denoted $\mathcal{B}(\exists \FO)$. Over finite words, $\mathcal{B}(\exists \FO)$-sentences
describe the celebrated class of \emph{piecewise testable languages} (see e.g. \cite{perrin1986first}).
In our quantitative setting, one could define for all $k \ge 0$
the class of linear combinations of the counting formulas from $\mathcal{B}(\exists \FO)_k$,
as we did for $\Ve_k$ (resp. $\We_k$) with $\MSO_k$ (resp. $\FO_k$).
While this class seems to be a good candidate for defining
``piecewise testable $\Rel$-polyregular functions'', it
does not admit a free variable minimization theorem depending on the
growth rate of the functions. Indeed, let $A \defined \{a,b\}$
and consider the indicator function $ \car{a A^*} = \card \phi$
for $\phi(x) \defined a(x) \land \forall y. y \ge x \in \mathcal{B}(\exists \FO)_1$.
Even if $|\mathbf{1}_{a A^*}(w)| = \bigO(1)$, this function cannot be written
as a linear combination of counting formulas from $\mathcal{B}(\exists \FO)_0$.
Indeed, if we assume the converse, then
$\car{a A^*}$ could be written $\sum_{i=1}^n \delta_i \car{L_i}$
for some piecewise testable languages $L_i$,
which implies that $a A^*$ would be piecewise testable, which is not the case.

%

\subparagraph*{Star-free $\Nat$-polyregular functions}
A very natural question is, given an
$\Nat$-polyregular function (recall that it is an element of
$\VeN \defined \Vect{\Nat}(\card \phi : \phi \in \MSO)$)
to decide whether it is in fact
a \emph{star-free $\Nat$-polyregular} function
(i.e. an element of $\WeN \defined \Vect{\Nat}(\card \phi : \phi \in \FO)$).
In this setting, we conjecture that
$\WeN = \VeN \cap \We$.
This question seems to be challenging. Indeed, the techniques introduced
in the current paper cannot directly be applied
 to solve it, since the residual automaton (see \cref{sec:aperiodic})
of an $\Nat$-polyregular function may need
labels which are not $\Nat$-polyregular, or even not nonnegative.
In other words, replacing the output group by an output monoid
seems to prevent from representing
the functions with  canonical objects
based on residuals.


\subparagraph*{Star-free $\Rel$-rational series}
In \cref{fig:Z-classes}, there is no generalization of the class $\We$
among the whole class of $\Rel$-rational series. 
We
are not
aware of a way to define a class of
``star free $\Rel$-rational series'', neither with logics nor
with $\Rel$-rational expressions. Indeed, allowing the use of Kleene star
for series automatically builds the whole class of
$\Rel$-rational series (including the indicator
functions of all regular languages).
 
From a logical standpoint, 
it is tempting to go from polynomial behaviors to
exponential ones by shifting from
first-order free variables to second-order free variables.
While this approach actually 
captures the whole class of $\Rel$-rational series,
it fails to circumscribe star-freeness. To make the above
statement precise, let us write 
$\MSO^X$ (resp. $\FO^X$) as the set
of $\MSO$ (resp. $\FO$) formulas with free
second-order variables, i.e. of the shape
$\varphi(X_1, \dots, X_k)$. Given
$\phi \in \MSO^X$, we let $\card \phi(w) \colon A^* \fonc \Rel$
be the function that counts second-order valuations.
As an example of the expressiveness of this model,
let us illustrate how to compute $w \mapsto (-2)^{|w|} \not \in \Ve$.

\begin{example}
    Let $\phi(X) \defined \top$, then $\card \varphi(w) = 2^{|w|}$.
    Let $\psi(X)$ be the first-order formula
    stating that $X$ contains the first position of the word,
    $X$ contains the last position of the word,
    and if $x \in X$, then $x+1 \not \in X$ and $x+2 \in X$.
    It is an easy check that $\card{\psi} = \Lodd$,
    even though $\psi \in \FO^X$ (but recall that
    $\Lodd$ is the indicator function 
    of a non star-free regular language).
    Now, $w \mapsto (-2)^{|w|}$
    equals 
    $\card{\phi} \times (2 \card{\psi} - 1)$.
\end{example}

We are now ready to explain formally how
both $\FO^X$ and $\MSO^X$ capture $\Rel$-rational series.

\begin{proposition}
    \label{prop:counting-mso-variables}
    For every function $f \colon A^* \to \Rel$,
    the following are equivalent:
    \begin{enumerate}
        \item \label{it:ratrat} $f$ is a $\Rel$-rational series;
        \item \label{it:MSOX} $f \in \Vect{\Rel}(\setof{\card \varphi}{\varphi \in \MSO^X})$;
        \item \label{it:FOX} $f \in \Vect{\Rel}(\setof{\card \varphi}{\varphi \in \FO^X})$.
    \end{enumerate}
\end{proposition}

In our setting, it seems natural to say that
$w \mapsto 2^{|w|}$ should be a star-free $\Rel$-rational series,
contrary to $w \mapsto (-2)^{|w|}$
(as observed in \cref{ex:gastin}, this approach contrasts with
the weighted logics of Droste and Gastin \cite{droste2019aperiodic},
for which $(-2)^{|w|}$ is considered as ``star free'').
Recall that in \cref{theo:Zsf-eigen}, we have characterized
$\We$ as the class of series whose
spectrum falls in $\{0,1\}$.
Following this result, we conjecture that a ``good'' notion
of star-free $\Rel$-rational series could be those
whose spectrum falls in the set $\Rea_+$ of
nonnegative real numbers. This way, exponential
growth is allowed (e.g. for $w \mapsto 2^{|w|}$)
but no periodic behaviors (e.g. for $w \mapsto (-2)^{|w|}$).

%

%% file: parts/appendix.tex
\section{Proofs of \cref{sec:prelim}}

\subsection{Proof of \cref{prop:precompose}}

In this section, we show that the functions of $\Ve_k$
are closed by precomposition under a regular function.
This proof is somehow classical and inspired by well-known composition techniques
for $\MSO$-transductions.

    \begin{definition}[Transduction]
      \label{def:mso-transduction}
        A ($k$-copying) \emph{$\MSO$-transduction} from $A^*$ to $B^*$ 
        consists in several $\MSO$ formulas over $A$:
        \begin{itemize}
            \item for all $1 \le j \le k$, a formula $\phi^{\Dom}_j(x) \in \MSO_1$;
            \item for all $1 \le j \le k$ and $a \in B$, a formula $\phi^a_j(x) \in \MSO_1$;
            \item for all $1 \le j,j' \le k$, a formula $\phi^{<}_{j,j'}(x,x') \in \MSO_2$.
        \end{itemize}
    \end{definition}

Let $w \in A^*$, we define the domain
$D(w) \defined \setof{(i,j)}{1 \le i \le |w|, 1 \le j \le k, w \models \phi^{\Dom}_j(i)}$.
Using the formulas $\phi^b_j(x)$ (resp. $\phi^{<}_{j,j'}(x,x')$), we can
label the elements of $D(w)$ with letters of $B$ (resp. define
a relation $<$ on the elements of $D(w)$). The transduction
is defined if and only if the structure $D(w)$ equipped with the
labels and $<$ is a word $v \in B^*$, for all $w \in A^*$. In this case,
the transduction computes the function that maps $w \in A^*$ to this $v \in B^*$.

It follows from \cite{engelfriet2001mso} that  regular functions
can (effectively) be described by $\MSO$-transductions.

\begin{claim}
    Let $\ell \ge 0$, $k \geq 1$,
    $\psi(x_1, \dots, x_{\ell}) \in \MSO_{\ell}$
    be a formula over $B$
    and $f \colon A^* \fonc B^*$ be computed by a $k$-copying
    $\MSO$-transduction.
    Let us write $W \defined \set{x_1, \dots, x_\ell}^{\set{1, \dots, k}}$.
    There exists formulas $\theta_{\rho} \in \MSO_\ell$
    over $A$ where $\rho$ ranges in $W$,
    such that
    for all $w \in A^*$,
    $\card \phi (f(w)) = \sum_{\rho \in W} \card{\theta_\rho}(w)$.
\end{claim}
\begin{proof}[Proof Sketch]
	Assume that the transduction is given by formulas 
	$\phi^{\Dom}_j(x)$, $\phi^a_j(x) \in \MSO_1$ for $a \in B$
	and $\phi^{<}_{j,j'}(x,x') \in \MSO_2$ as in \cref{def:mso-transduction}.
    Let $\psi$ be an $\MSO$ formula
    over $B$ with first order variables $x_1, \dots, x_\ell$
    and second order variables $(X_1, \dots, X_k),
    (Y_1, \dots, Y_k), \dots $.
    Let 
    $\rho$ be a mapping from
    $\set{x_1, \dots, x_\ell}$ to $\set{1, \dots, k}$.
    We define by induction on $\psi$ the formula $\psi_{\rho}$
    as follows (it roughly translates the formula
    from $B$ to $A$ using the transduction):
    \begin{align*}
        (\exists x. \phi)_\rho &\defined
        \bigvee_{j = 1}^k
        \exists x. \phi^{\Dom}_j(x) \wedge \phi_{\rho+[x \mapsto j]} \\
        (\exists X. \phi)_{\rho} &\defined
        \exists X_1, \dots, X_k.
            \bigwedge_{j = 1}^k (\forall x \in X_j, \phi^{\Dom}_j(x))
            \wedge \phi_{\rho} \\
        (\neg \phi)_\rho &\defined \neg (\phi_\rho) \\
        (\phi \vee \phi')_{\rho} &\defined \phi_\rho \vee \phi'_{\rho} \\
        (P_a(x))_{\rho} &\defined \phi_{\rho(x)}^a(x) \\
        (x < y)_{\rho} &\defined \phi_{\rho(x), \rho(y)}^{<}(x,y). \\
        (x \in X)_{\rho} &\defined \bigvee_{j = 1}^k
                \phi^{\Dom}_j(x) \wedge (x \in X_j)
    \end{align*}
    \newcommand{\pos}{\small \operatorname{\textsf{pos}}}%
    It is then a mechanical check that the translation 
    works as expected. In the following equation,
    we fix $w \in A^*$ and we let $\pos \colon D(w) \fonc [1{:}|f(w)|]$
    be the function that maps a tuple $(i,j)$ to the corresponding
    position in the word $f(w) \in B^*$.
    To simplify notations, given $\rho \in W$, a word $w \in A^*$,
    and a valuation
    $\tau \colon \set{x_1, \dots, x_\ell} \to [1{:}|w|]$,
    we write $\pos[\tau \times \rho](\vec{x})
    \defined \pos (\tau(x_1), \rho(x_1)), 
    \dots,
    \pos (\tau(x_\ell), \rho(x_\ell))$.
    \begin{align*}
        \card{\phi}(f(w)) 
        &= 
        \card{\setof{
                \nu \colon \set{x_1, \dots, x_\ell} \to [1{:}|f(w)|]
        }{
            f(w) \models \psi(\nu(x_1), \dots, \nu(x_{\ell}))
        }} \\
        &= \sum_{\rho \in W}
        \card{\setof{
                \tau \colon \set{x_1, \dots, x_\ell} \to [1{:}|w|]
        }{
            f(w) \models \psi(\pos[\tau \times \rho](\vec{x}))
        }} \\
        &= \sum_{\rho \in W}
        \card{\setof{
                \nu \colon \set{x_1, \dots, x_\ell} \to \set{1, \dots, |w|}
        }{
            w \models \psi_\rho(\nu)
            \wedge \bigwedge_{i = 1}^\ell \phi^{\Dom}_{\rho(x_i)}(x_i)
        }} \\
    \end{align*}
    We then let $\theta_\rho \defined \psi_\rho
            \wedge \bigwedge_{i = 1}^\ell \phi^{\Dom}_{\rho(x_i)}(x_i)$
    to conclude.
\end{proof}

The result follows immediately since $ \Ve_{\ell}$ is closed
under taking sums and $\Rel$-external products.

\subsection{Proof of \cref{prop:polypoly}}

    We first show that any $\Rel$-polyregular
    function can be written under the form $\polysum \circ g$
    where $g : A^* \fonc \{\pm 1\}^*$ is polyregular.
    This is an immediate consequence of the following claims.

	\begin{claim}
        For all $\varphi \in \MSO$, there exists a polyregular function 
        $f \colon A^* \to \set{\pm 1}^*$ such that  $\card \varphi = \polysum \circ
        f$.
	\end{claim}
    \begin{proof}
        Polyregular functions are characterized
        in  \cite[Theorem 7]{bojanczyk2019string}
        as the functions computed by 
        (multidimensional) $\MSO$-interpretations.
        Recall that an $\MSO$-interpretation
        of dimension $k \in \Nat$
        is given by a formula $\varphi_{\leq}(\vec{x},\vec{y})$ 
        defining a total ordering
        over $k$-tuples of positions,
        a formula $\varphi^{\Dom}(\vec{x})$ that selects
        valid positions,
        and formulas $\varphi^{a}(\vec{x})$ that
        place the letters over the output word
        \cite[Definition 1 and 2]{bojanczyk2019string}.
        In our specific situation, letting
        $\varphi_{\leq}$ be the usual lexicographic ordering of positions 
        (which is $\MSO$-definable)
        and placing the letter $1$ over every element of the output
        is enough: the only thing left to do is select
        enough positions of the output word.
        For that, we let $\varphi^{\Dom}$ be defined as $\varphi$ itself.
        It is an easy check that this $\MSO$-interpretation
        precisely computes $1^{f(w)}$ over $w$,
        hence computes $f$ when post-composed with $\polysum$.
    \end{proof}
	
	\begin{claim}
        The set $\setof{\polysum \circ f}{ 
        f \colon A^* \to \set{\pm 1}^* \text{ polyregular}}$
        is closed under sums and external $\Rel$-products.
	\end{claim}
    \begin{proof}
        Notice that $\polysum \circ f + \polysum \circ g = \polysum \circ (f
        \cdot g)$ where $f \cdot g(w) \defined f(w) \cdot g(w)$.
        As polyregular functions are closed under concatenation
        \cite{bojanczyk2018polyregular},
        the set of interest is closed under sums.
        To prove that it is closed under external $\Rel$-products,
        it suffices to show that it is closed under negation.
        This follows because one can permute the
        $1$ and $-1$ in the output of a polyregular function
        (polyregular functions are closed under post-composition
        by a morphism).
    \end{proof}

    Let us consider a \kl{polyregular function}
	$g \colon A^* \to \set{ \pm 1 }^*$.
	The maps $g_+ \colon w \mapsto |g(w)|_1$
	and $g_- \colon w \mapsto |g(w)|_{-1}$
    are \kl{polyregular functions with unary output}
	(since they correspond to a post-composition by
	the regular function which removes some letter, and polyregular functions are
	closed under post-composition by a regular function \cite{bojanczyk2018polyregular}).
	Hence $g_-$ and $g_+$ are polyregular functions with unary output,
    a.k.a. $\Nat$-polyregular functions.
    As a consequence,
	${\polysum} \circ g = g_+ - g_-$ lies in $\Ve$.

\section{Proofs of \cref{ssec:expressions}}

\subsection{Proof of \cref{claim:vks:cauchyincrease}}

Let $f \in \Ve_{k}$  and $g \in \Ve_{\ell}$, we
(effectively) show that  $f \cauchy g \in \Ve_{k+\ell+1}$.

First, observe that if $f,g,h : A^* \fonc \Rel$ and $\gamma, \delta \in \Rel$, then
$(\gamma f+ \delta g) \cauchy h = \gamma (f \cauchy g) + \delta (g \cauchy h)$.
Thus it is sufficient to show the result for $f = \card \phi$ and $g = \card \psi$
with $\phi (x_1, \dots, x_k) \in \MSO_{k}$ and $\psi(y_1, \dots, y_{\ell}) \in \MSO_{\ell}$.
For all $w \in A^*$ we have:
	\begin{align*}
		(\card \phi \cauchy \card \psi) (w)
		 & =
		\sum_{0 \leq i \leq |w|}
		\sum_{i_1, \dots, i_k \leq i}
		\sum_{j_1, \dots, j_{\ell} > i}
		\mathbf{1}_{w[1{:}i] \models \phi(i_1, \dots, i_k)}
		\times
		\mathbf{1}_{w[i{+}1:|w|] \models \psi'(j_1, \dots, j_{\ell})}                             \\
		 & =
         \card{\phi}(\movi) \cdot \card{\psi}(w)  \\
         &+
		\sum_{1 \leq i \leq |w|}
		\sum_{i_1, \dots, i_k \leq i}
		\sum_{j_1, \dots, j_{\ell} > i}
		\mathbf{1}_{w[1{:}i] \models \phi(i_1, \dots, i_k)}
		\times
		\mathbf{1}_{w[i{+}1:|w|] \models \psi'(j_1, \dots, j_{\ell})}                             \\
		 & = 
         \card{\phi}(\movi) \cdot \card{\psi}(w) 
         + \card (\phi'(z,x_1, \dots, x_k) \wedge \psi'(z,y_1, \dots, y_l))
		(w) \qedhere
	\end{align*}
where $\phi'(z, x_1, \dots, x_k) \in \MSO_{k+1} $ is a formula such that $w \models \phi'(i,i_1, \dots, i_k)$
if and only if $i_1, \dots, i_k \le i$ and $w[1{:}i] \models \phi(i_1, \dots, i_k$)
(this is a regular property which is $\MSO$ definable), and similarly
for $\psi'$.

\subsection{Proof of \cref{lem:vks:inductcauchy}}

	Let $k \geq 0$, we want to show that
	$\Ve_{k+1} = \Vect\Rel{{}}(\setof{\mathbf{1}_{L} \cauchy f}{L \text{\normalfont~regular}, f \in \Ve_k})$.
    Observe that for all
    $f \colon A^* \to \Rel$,
    $\car{\set{\movi}} \cauchy f$ equals $f$,
    therefore 
    $\Ve_k \subseteq \Vect\Rel{{}}(\setof{\mathbf{1}_{L} \cauchy f}{L \text{\normalfont~regular}, f \in \Ve_k})$.
	As in the proof of  \cref{claim:vks:cauchyincrease}, it is sufficient
	to show that $\card \phi$ for $\phi(x_1, \dots, x_{k+1}) \in \MSO_{k+1}$,
	can be written as a linear combination
	of $\car{L} \cauchy f$ where $L$ is a regular language.
	Observe that for all $w \in A^+$, for
    all valuation $i_1, \dots, i_k$ of $x_1, \dots, x_k$,
    we can define $P \defined \setof{1 \le j \le k}{i_j = \min \{i_1, \dots, i_k\}}$
    (i.e. the $x_j$ for $j \in P$ are the variables with minimal value).
    Therefore, for all $w \in A^+$:
\begin{equation*}
\begin{aligned}
\card \varphi (w)
	=
    \sum_{\emptyset \subsetneq P \subseteq [1{:}k]}
	\sum_{w = uv, u \neq \movi}
    \card (\varphi \wedge \bigwedge_{j \in P} x_j = |u| \wedge 
        \bigwedge_{j
    \not \in P} x_j >
		|u|) (w).
\end{aligned}
\end{equation*}

It is an easy check that one can
(effectively) build a regular language $L^P \subseteq A^+$
and a formula $\psi^P$
such that
for all $u \in A^+$, $v \in A^*$,
$uv \models \varphi \wedge \bigwedge_{j \in P} (x_j = |u|) \wedge
(\bigwedge_{j \not \in P} x_j >
	|u|)$
if and only if
$u \in L^P$ and $v \models \psi^P( (x_j)_{j \not \in P} )$.
Thus, for all $w \in A^+$:

\begin{align*}
\card \varphi (w)
    &=
    \sum_{\emptyset \subsetneq P \subseteq [1{:}k]}
	\sum_{w = uv}
    \car{L^P} (u) \times \card{\psi^P}(v)
    \\
    &= 
    \underbrace{
    \sum_{\emptyset \subsetneq P \subseteq [1{:}k]}
    (\car{L^P} \cauchy \card{\psi^P})}_{\defined g} (w)
    \quad .
\end{align*}

Notice that $\psi^P$ has exactly
$k - |P| \leq k - 1$ free-variables, thus 
$g$ belongs to $\Vect\Rel{{}}(\setof{\mathbf{1}_{L} \cauchy f}{L \text{\normalfont~regular}, f \in \Ve_k})$.
Observe moreover that $g(\movi) = 0 = \card{\varphi}(\movi)$ because
$k+1 > 0$.

\section{Proofs of \cref{ssec:matrix}}

\subsection{Proof of \cref{lem:capturing-eigenvalues}}

	 Let $f \colon A^* \to \Rel$ be a $\Rel$-rational series and
	$(I, \mu, F)$ be a minimal $\Rel$-linear representation of $f$ of
         dimension $n$. First note that $(I, \mu, F)$
	is also a minimal $\Rat$-linear representation of $f$
	by \cite[Theorem 1.1 p 121]{berstel2011noncommutative}
	($\Rat$-linear representations are defined by allowing
	rational coefficients whithin the matrices and vectors, instead of integers).
	Let $w \in A^*$, $\lambda \in \Spec(\mu(w))$
	and consider a complex eigenvector $V \in \Mat{n,1}(\Com)$
	associated to $\lambda$. We let $||V|| \defined {}^t V V$,
	observe that it is a positive real number.
	Because $(I,\mu,F)$ is a minimal $\Rat$-linear
	representation of $f$, then
	$\Vect{\Rat}(\setof{\mu(u) F}{u \in A^*}) = \Rat^n$
	by \cite[Proposition~2.1 p~32]{berstel2011noncommutative}.
	Hence there exists numbers $\alpha_j \in \mathbb{C}$
	and words $u_j \in A^*$
	such that $V = \sum_{j = 1}^n \alpha_j \mu(u_j) F$.
	Symmetrically by \cite[Proposition 2.1 p 32]{berstel2011noncommutative},
	there exists numbers $\beta_i \in \mathbb{C}$ and words $v_i \in A^*$
	such that ${}^t V = \sum_{i = 1}^n \beta_i I \mu(v_i)$.
	Therefore:
	\begin{align*}
		\lambda^X ||V||      = {}^t V \mu(w)^X V &          
		             = \sum_{i ,j= 1}^{n}
		\alpha_i \beta_j  I \mu(v_i w^X u_j) F              
		             = \sum_{i ,j= 1}^{n}
		\alpha_i \beta_j  f(v_i w^X u_j).
        \qedhere
	\end{align*}
	The result follows since $||V|| \neq 0$ (it is an eigenvector).

\subsection{Proof of \cref{lem:Zpoly-pump-first}}

\label{pro:pump}

	If $L$ is a regular language, the fact that
	$\car{L}$ is $N$-polynomial for some $N \ge 0$
	follows from the traditional pumping lemmas.
	Now let $f, g \colon A^* \to \Rel$ be respectively ultimately
	$N_1$-polynomial and ultimately $N_2$-polynomial.
	The fact that $f + g$ and $\delta f$ for $\delta \in \Rel$ are
	ultimately $(N_1 \times N_2)$-polynomial is obvious.
	In the rest of \cref{pro:pump}, we focus on the main
	difficulty which is the Cauchy product of two functions.
	For that, we will first prove the following claim
	about Cauchy products of polynomials.

    \begin{claim}
        \label{claim:sum-is-polynomial}
        For every $p \in \Nat$, $\sum_{i = 0}^X i^p$ is a polynomial in $X$.
    \end{claim}
    \begin{proof}
        It is a folklore result, but let us prove it using finite differences.
        If $f \colon \Nat \to \Rat$, let $\Delta f \colon n \mapsto f(n+1) - f(n)$.
        Let us now prove by induction that every function $f \colon \Nat \to
        \Rat$ such that $\Delta^{p} f = 0$ for some $p \ge 1$
        is a polynomial. For $p = 1$, this holds because
        $f$ must be constant.
        For $p + 1 > 1$, 
        if we assume that $\Delta^{p+1} f = 0$,
        then $\Delta^{p} f$ is a constant $C$. Let 
        $g \defined f - C \frac{n^p}{p!}$, 
        and remark that $\Delta^{p} g = 0$.
        By induction hypothesis $g$ is a polynomial, hence so is $f$.

        Finally, a simple induction proves that $\Delta^{p+2} (X \mapsto \sum_{i = 0}^X i^p) = 0$.
    \end{proof}

    \begin{claim}
        \label{claim:cauchy-product-polynomials}
        Let $P, Q \in \Rat[X, Y_1, \dots, Y_\ell]$ be
        two multivariate polynomials,
        then their Cauchy product
        $P \cauchy Q (X, Y_1, \dots, Y_\ell)
        \defined \sum_{i = 0}^{X} P(i, Y_1, \dots, Y_\ell) Q(Y - i, Y_1, \dots,
        Y_\ell)$
        belongs to $\Rat[X, Y_1, \dots, Y_\ell]$.
    \end{claim}
    \begin{proof}
        By linearity of the Cauchy product, it suffices
        to check that the result holds 
        for products of the form
        $(X^p Y_1^{p_1} \cdots Y_\ell^{p_\ell})
        \cauchy 
        (X^q Y_1^{q_1} \cdots Y_\ell^{q_\ell})
        = (X^p \cauchy X^q) \times 
        Y_1^{p_1} \cdots Y_\ell^{p_\ell}
        Y_1^{q_1} \cdots Y_\ell^{q_\ell}$.
        Hence, the only thing left to check
        is that $X^p \cauchy X^q$ is a polynomial in $X$.

        \begin{align*}
        X^p \cauchy X^q (Y)
        &= \sum_{i = 0}^Y i^p (Y - i)^q \\
        &= \sum_{i = 0}^Y i^p \sum_{k = 0}^q \binom{q}{k} Y^k (-i)^{q - k} \\
        &= \sum_{k = 0}^q \binom{q}{k} Y^k \sum_{i = 0}^Y i^p (-i)^{q - k} \\
        &= \sum_{k = 0}^q \binom{q}{k}  (-1)^{q - k} Y^k \sum_{i = 0}^Y i^{p + q - k}
        \end{align*}

        Which is a polynomial thanks to
        \cref{claim:sum-is-polynomial}.
    \end{proof}

    Let us now prove that $f \cauchy g$ is ultimately
    $N \defined(N_1\times N_2)$-polynomial.
    For that, let us consider
    $\alpha_0, u_1, \alpha_1, \dots, u_\ell, \alpha_\ell \in A^*$
    and prove that 
    $(f \cauchy g) (\alpha_0 u_1^{N X_1} \alpha_1 \cdots u_{\ell}^{N
    X_{\ell}}\alpha_{\ell})$
    is a polynomial for $X_1, \dots, X_{\ell}$ large enough.

    \begin{equation*}
        \begin{aligned}
             (f \cauchy g) (\alpha_0 u_1^{N X_1} \alpha_1 \cdots u_{\ell}^{N X_{\ell}}\alpha_{\ell})
            & = f(\alpha_0 u_1^{N X_1} \alpha_1 \cdots
            u_{\ell}^{N X_{\ell}}\alpha_{\ell})g(\movi) \\
             & + \sum_{j=0}^{\ell} \sum_{i=0}^{|\alpha_j|-1}
            f(\alpha_0 u_1^{N X_1} \alpha_1 \cdots u_j^{N X_j}(\alpha_{j}[1{:}i]))\\
            &\times g((\alpha_{j}[i{+}1{:}|\alpha_j|]) u_{j+1}^{N X_{j+1}}\cdots \alpha_{\ell})        \\
             & + \sum_{j=1}^{\ell}  \sum_{i=0}^{|u_j^N|-1} \sum_{Y=0}^{X_{j}-1}
             f(\alpha_0 u_1^{N X_1} \alpha_1 \cdots u_j^{N Y} (u_j^N[1{:}i]))
             \times
             \\
             &
            \quad \quad g((u_j^N[i{+}1{:}|u_j^N|])u_j^{N(X_j-Y-1)} \cdots \alpha_{\ell})
        \end{aligned}
    \end{equation*}

    From the hypothesis on $f$, we deduce that the first term of this sum
    is ultimately $N_1$-polynomial, hence ultimately $N$-polynomial.
    We conclude similarly for the second term of this sum,
    because the product of two polynomials is a polynomial.

    Let us now focus on the third term.
    Using the induction hypotheses on $f$ and $g$,
    there exists polynomials
    $P_{j,i}$ and $Q_{j,i}$ such that the following equalities
    ultimately hold,
    where $(X_1, \dots, \hat{X_j}, \dots X_\ell)$
    denotes the tuple obtained by removing the $j$-th element
    from $(X_1, \dots, X_\ell)$:
    \begin{align*}
        f(\alpha_0 u_1^{N X_1} \alpha_1 \cdots u_j^{N Y} (u_j^N[1{:}i]))
        &=
        P_{j,i}(Y, X_1, \dots, \hat{X_j}, \dots X_\ell)
        \\
        g((u_j^N[i{+}1{:}|u_j^N|])u_j^{N(X_j-Y-1)} \cdots \alpha_{\ell})
        &= Q_{j,i}(Y, X_1, \dots, \hat{X_j}, \dots X_\ell)
    \end{align*}

    As a consequence, we can rewrite the third term as a Cauchy product
    of polynomials for large enough values of $X_1, \dots, X_\ell$:
    \begin{align*}
             & \sum_{j=1}^{\ell}  \sum_{i=0}^{|u_j^N|-1} \sum_{Y=0}^{X_{j}-1}
             f(\alpha_0 u_1^{N X_1} \alpha_1 \cdots u_j^{N Y} (u_j^N[1{:}i]))
            g((u_j^N[i{+}1{:}|u_j^N|])u_j^{N(X_j-Y-1)} \cdots \alpha_{\ell}) \\
             &= 
             \sum_{j=1}^{\ell}  \sum_{i=0}^{|u_j|-1} \sum_{Y=0}^{X_{j}-1}
             P_{j,i} (Y, X_1, \dots, \hat{X_j}, \dots, X_\ell)
             Q_{j,i} (X_j - Y - 1, X_1, \dots, \hat{X_j}, \dots, X_\ell) \\
             &= 
             \sum_{j=1}^{\ell}  \sum_{i=0}^{|u_j|-1} 
             P_{i,j} \cauchy Q_{j,i} (X_j - 1)
    \end{align*}

    Thanks to \cref{claim:cauchy-product-polynomials}, we conclude that
    this third term is also ultimately a polynomial.

\section{Proofs of \cref{sec:pebblemin}}

\subsection{Proof of \cref{claim:skel:skeletons-totally-ordered}}

    First of all, given a leaf $x \in \Leaves(F)$,
    $\Ske(x) = \{ x \}$ contains $x$. Hence, every leaf
    is contained in at least one skeleton. It remains
    to show that if  $\nod$ and $\nod'$ are two nodes such that
    $x \in \Ske(\nod)$ and $x \in \Ske(\nod')$,
    then $ \Ske(\nod) \subseteq \Ske(\nod')$
    or the converse holds.
    
	As $\Ske(\nod)$ contains only children of $\nod$, one
	deduces that $x$ is a children of both $\nod$ and $\nod'$.
	Because $F$ is a tree, parents of $x$ are totally ordered by their
	height in the tree. As a consequence, without loss of generality,
	one can assume that $\nod$ is a parent of $\nod'$.
	Because $\Ske(\nod)$ is a subforest of $F$ containing $x$,
	it must contain $\nod'$. Now, by definition of skeletons,
	it is easy to see that whenever $\nod' \in \Ske(\nod)$, we have
	$\Ske(\nod') \subseteq \Ske(\nod)$.




\subsection{Proof of \cref{rem:skel:bound-depnodes}}

	Let $x \in \Leaves(F)$, we show that the number of
	$x'$ such that $x'\dependson x$ is bounded (independently
	from $x$ and $F \in \forests{\mu}{d}$).
	Observe that $\partitionskel(x')$ is either an ancestor
	or the sibling of an ancestor of $\partitionskel(x)$.
	Observe that for all $\nod \in \Nodes(F)$, $\Ske(\nod)$ is a
	binary tree of height at most $d$, thus is has at most
	$2^d$ leaves.
	Moreover, $\partitionskel(x)$ has at most
	$d$ ancestors and $2d$ immediate siblings of its ancestors.
	As a consequence, there are at most
    $3d \times 2^d$ leaves that depend on $x$.

\subsection{Proof of \cref{lem:skel:inv-decomp}}

	Let $d \ge 0$, $M$ be a finite monoid,
    $\mu \colon A^* \to M$, $k \geq 1$, and $\psi \in \INV_k$.
    We want to build a function $g \in \Ve_{k-1}$
    such that for every $F \in \forests{\mu}{d}$,
	$g(F) = \card (\psi(\vec{x}) \wedge \msodep(\vec{x}))(F)$
	(since $\forests{\mu}{d}$ is a regular language of
	$\hat{A}^*$, it does not matter how $g$ is defined
	on inputs $F \not \in \forests{\mu}{d}$).

First, we use the lexicographic order to find the first pair $(x_i,x_j)$ that is
dependent in the tuple $\vec{x}$. This allows to partition our set of
valuations as follows:
\begin{align*}
	 & \setof{\vec{x} \in \Leaves(F)}{F, \vec{x} \models \psi \wedge
		\msodep(\vec{x})}                                                \\
	 & = \biguplus_{1 \leq i < j \leq n}
	\setof{\vec{x} \in \Leaves(F)}{F, \vec{x} \models \psi \wedge
		\msodep(x_i, x_j) \wedge
		\bigwedge_{(k,{\ell}) \lexl (i,j)} \neg \msodep(x_k, x_{\ell})
	}
	\\
	 & = \biguplus_{1 \leq i < j \leq n}
	\setof{\vec{x} \in \Leaves(F)}{F, \vec{x} \models
		\underbrace{\psi \wedge
			x_j \operatorname{\mathsf{depends}-\mathsf{on}} x_i \wedge
			\bigwedge_{(k,{\ell}) \lexl (i,j)} \neg \msodep(x_k, x_{\ell})}_{
			\defined
			\psi_{i \to j}(\vec{x})
		}
	}                                                                \\
	 & \quad\quad\quad\quad \cup
	\setof{\vec{x} \in \Leaves(F)}{F, \vec{x} \models \underbrace{
			\psi \wedge
			x_i \operatorname{\mathsf{depends}-\mathsf{on}} x_j \wedge
			\bigwedge_{(k,{\ell}) \lexl (i,j)} \neg \msodep(x_k, x_{\ell})
		}_{\defined \psi_{i \leftarrow j}(\vec{x})}
	}
\end{align*}
As a consequence,
$\card{(\psi \wedge \msodep)}
	=
	\sum_{1 \leq i < j \leq n}
	\card{\psi_{i \to j}}
	+ \card{\psi_{i \leftarrow j}}
	- \card{\psi_{i \to j} \wedge \psi_{i \leftarrow j}}
$
(the last term removes the cases when 
both $x_i \dependson x_j$ and  $x_j \dependson x_i$,
which occurs e.g. when $x_i = x_j$).

We can now rewrite this sum using $\intro\existsexactly{\ell} x_j. \psi$
to denote the fact that there exists exactly $\ell$ different values for $x$
so that $\psi(\dots, x_j, \dots)$ holds (this quantifier
is expressible in $\MSO$ at every fixed $\ell$).
Thanks to \cref{rem:skel:bound-depnodes}, there exists a bound
$N_d$ over the maximal number of leaves that dependent on a leaf $x_i$
(among forests of depth at most $d$.) Hence: 

\begin{align*}
	\card{(\psi \wedge \msodep)}
    &=
	\sum_{1 \leq i < j \leq n}
	\card{\psi_{i \to j}}
	+ \card{\psi_{i \leftarrow j}}
	- \card{\psi_{i \to j} \wedge \psi_{i \leftarrow j}}
    \\
    &=
	\sum_{1 \leq i < j \leq n}
	\sum_{0 \leq \ell \leq N_d} \ell \cdot \card{\exists^{=\ell} x_j. \psi_{i \to j}}
    \\
    &+
	\sum_{1 \leq i < j \leq n}
    \sum_{0 \leq \ell \leq N_d} \ell \cdot \card{\exists^{=\ell} x_i. \psi_{i
    \leftarrow j}}
    \\
    &-
	\sum_{1 \leq i < j \leq n}
    \sum_{0 \leq \ell \leq N_d} \ell \cdot \card{\exists^{=\ell} x_i. \psi_{i \to j} \wedge \psi_{i \leftarrow j}}
\end{align*}

\subsection{Proof of \cref{lem:skel:indep-iter-zero}}

In order to prove \cref{lem:skel:indep-iter-zero},
we consider $f$ such that $\findep \neq 0$.
Our goal is to construct a pumping family to exhibit a
growth rate of $\findep$.
To construct such a pumping family, we will rely on the
fact that independent tuples of leaves have a very specific behavior
with respect to the factorization forest.
Given a node $\nod$, we write
$\intro\skelstart(\nod) \defined \min_{\leq} \{ y \in \Leaves(F) \cap
	\Ske(\nod) \}$
and $\intro\skelend(\nod) \defined \max_{\leq} \{ y \in \Leaves(F) \cap \Ske(\nod)
	\}$.

\begin{claim}
    Let $x_1, \dots, x_k$ be an independent tuple of $k \geq 1$ leaves
    in a forest $F \in \forests{\mu}{d}$ factorizing a word $w$.
    Let $\vec{\nod}$ be the vector of nodes such that
    $\nod_i \defined \partitionskel(x_i)$ for all $1 \le i \le k$.
    One can order the $\nod_i$
    according to their position in the word $w$
    so that
    $1 < \skelstart(\nod_1) \leq \skelend(\nod_1) < \dots < \skelstart(\nod_k)
	\leq \skelstart(\nod_k) < |w|$.
\end{claim}
\begin{claimproof}
    Assume by contradiction that there exists a pair $i < j$
    such that $\skelstart(\nod_j) \geq \skelend(\nod_i)$.
    We
    then 
    know that $\skelstart(\nod_i) \leq \skelstart(\nod_j) \leq
    \skelend(\nod_i)$.
    In particular, 
    $\partitionskel(\skelstart(\nod_i)) = \nod_i$
    is an ancestor of $\skelstart(\nod_j)$,
    hence $\nod_i$ is an ancestor of $\nod_j$.
    This contradicts the independence of $\vec{x}$.

    Assume by contradiction that there exists $i$ such that
    $\skelstart(\nod_i) = 1$ (resp. $\skelend(\nod_i) = |w|)$.
    Then $\partitionskel(x_i)$ must be the root of $F$, but
    then $\vec{x}$ cannot be an independent tuple.
\end{claimproof}

Given an independent tuple $x_1, \dots, x_k \in \Leaves(F)$, with
$\partitionskel(\vec{x}) = \vec{\nod}$, ordered by their
position in the word, let us define
$m_0 \defined \mu(w[1{:}\skelstart(\nod_1) {-}1])$,
$m_k \defined \mu(w[\skelend(\nod_k) {+} 1 {:} w_{|w|}])$
and
$m_i \defined \mu(w[\skelend(\nod_k) {+} 1{:}\skelstart(\nod_{i+1}) {-} 1])$
for $1 \leq i \leq k-1$.

\begin{definition}[Type of a tuple of $\partitionskel$]
	Let $F \in \forests{\mu}{d}$ factorizing a word $w$,
    $\vec{x}$ be an independent tuple of leaves in $F$,
    and $\vec{\nod} = \partitionskel(\vec{x})$.
	Without loss of generality assume that the nodes
	are ordered by $\skelstart$.
	The \emph{type} $\skeltype(\vec{\nod})$
	in the forest $F$ is
	defined as
	the tuple $(m_0, \Ske(\nod_1), m_1, \dots, m_{k-1},
		\Ske(\nod_k), m_k)$.
\end{definition}

At depth $d$, there are finitely many possible types
for tuples of $k$ nodes,
which we collect in the set $\Types_{d,k}$.
Moreover, given a type $T \in \Types_{d,k}$, one can build
the $\MSO$ formula $\skelhastype{T}(\vec{\nod})$ over $\forests{\mu}{d}$
that tests whether a tuple of \emph{nodes} $\vec{\nod}$
is of type $T$, and can be obtained
as $\partitionskel(\vec{x})$ for some tuple $\vec{x}$ of independent
leaves.
The key property of types
is that counting types is enough to count independent
valuations for a formula $\psi \in \INV$.

\begin{claim}
	\label{app:claim:bijection-types}
    Let $k \geq 1$, $d \geq 0$, $M$ be a finite monoid, $\mu \colon A^* \to M$
    be a morphism.
	Let $T \in \Types_{d,k}$, $F \in \forests{\mu}{d}$,
	$\vec{x}$ and $\vec{y}$ be two $k$-tuples of independent leaves
	of $F$ such that $\skeltype(\partitionskel(x_1), \dots,
		\partitionskel(x_k)) =
		\skeltype(\partitionskel(y_1), \dots,
		\partitionskel(y_k)) = T$.
		\\
	There exists a bijection $\sigma
		\colon L_1 \to L_2$, where
	$L_1 \defined \Leaves(F) \cap \bigcup_{i = 1}^k \Ske(\partitionskel(x_i))$
	and
	$L_2 \defined \Leaves(F) \cap \bigcup_{i = 1}^k \Ske(\partitionskel(y_i))$,
	such that
	for every $z \in L_1^k$,
	for every formula $\psi \in \INV_k$,
	$F \models \psi(z)$
	if and only if $F \models \psi(\sigma(z))$.
\end{claim}
\begin{proof}[Proof Sketch]
	Because of the type equality, we know that
	$\Ske(\partitionskel(x_i))$ and $\Ske(\partitionskel(y_{i}))$
	are isomorphic for $1 \leq i \leq k$. As the skeletons are disjoint
	in an independent tuple, this automatically provides the
	desired bijection $\sigma$.

	Let us now prove that $\sigma$ preserves the semantics of
	invariant formulas. Notice that this property is stable under
	disjunction, conjunction and negation. Hence, it suffices to check
	the property for the following
	three formulas $\msoprod_m(x,y)$,
	$\msoleft_m(x)$, $\msoright_m(y)$ and $\msoisleaf(x)$.
    For $\msoisleaf$, the result is the consequence of
    the fact that $\sigma$ sends leaves to leaves.

    Let us prove the result for $\msoprod_m$ and leave the other
    and leave the other
    cases as an exercise.
    Let $(y,z) \in L_1^2$.
    By definition of $L_1$, 
    there exists $1 \leq i,j \leq k$
    such that $y \in \Leaves(F) \cap \Ske(\partitionskel(x_i))$
    and $z \in \Leaves(F) \cap \Ske(\partitionskel(x_j))$.
    To simplify the argument, let us assume that
    $y < z$ and $i + 1 =  j$.
    Let $w \defined \repforest(F)$,
    and $m_{y,z} \defined \mu(w[y:z])$.
    One can decompose the computation of $m_{y,z}$
    as follows:
    \begin{align*}
        m_{y,z}
                    &= \mu(w[y:z])
                    \\
                    &=  \mu(w[y:\skelend(x_i)]w[\skelend(x_i) + 1: \skelstart(x_{i+1}) - 1]
        w[\skelstart(x_{i+1}):z])
        \\
                    &=
        \mu(w[y:\skelend(x_i)]) m_i \mu(w[\skelstart(x_i):z])
    \end{align*}

    Therefore,
    $\mu(w[y:z])$ only depends on
    $\Ske(\partitionskel(y)) = \Ske(\partitionskel(x_i))$,
    the position of $y$ in $\Ske(\partitionskel(y))$,
    $\Ske(\partitionskel(z)) = \Ske(\partitionskel(x_{i+1}))$,
    the position of $z$ in $\Ske(\partitionskel(z))$,
    and
    $m_i$, all of which are presreved by the bijection $\sigma$.
    Hence,
    $\mu(w[y:z]) = \mu(w[\sigma(y) : \sigma(z)])$.
    Therefore, 
    $F \models \msoprod_m(y,z)$ if and only if
    $F \models \msoprod_m(\sigma(y),\sigma(z))$.

    It is an easy check that a similar argument works
    when $j \neq i+1$.
\end{proof}

Now, we show that counting the valuations of a $\INV$
formula can be done by counting the number of tuples
of each type.

\begin{lemma}
	\label{lem:skel:decomp-indep-type}
    Let $k \geq 1$, $d \geq 0$, $M$ be a finite monoid, $\mu \colon A^* \to M$
    be a morphism.
    For every $\psi \in \INV_k$,
    there exists computable coefficients $\lambda_T \ge 0$,
	such that the following functions from $\forests{\mu}{d}$
    to $\Nat$ are equal:
	\begin{equation*}
		{\card \psi}_{\mathsf{indep}}
		\defined 
		\card{(\psi \wedge \neg \msodep)}
		=
		\sum_{T \in \Types_{d,k}}
		\lambda_T
        \cdot
		\card{\skelhastype{T} }
	\end{equation*}
\end{lemma}
\begin{proof}

Using the claim, we can now proceed to prove \cref{lem:skel:decomp-indep-type}.
	\begin{align*}
		\card{\psi \wedge \neg \msodep}(F)
		 & =
		\sum_{\vec{x} \text{ indep}}
        \car{F \models \psi(\vec{x})}
		\\
		 & =
		\sum_{T \in \Types_{d,k}}
		\sum_{\vec{\nod} \in \Nodes(F)}
		\sum_{\vec{x} \text{ indep}}
        \car{F \models \psi(\vec{x})}
		\car{\vec{\nod} = \partitionskel(\vec{x})}
		\car{\skelhastype{T}(\vec{\nod})}
		\\
		 & =
		\sum_{T \in \Types_{d,k}}
		\sum_{\vec{\nod} \in \Nodes(F)}
		\car{\skelhastype{T}(\vec{\nod})}
		\left(
		\sum_{\vec{x} \text{ indep}}
        \car{F \models \psi(\vec{x})}
		\car{\vec{\nod} = \partitionskel(\vec{x})}
		\right)
		\\
		 & =
		\sum_{T \in \Types_{d,k}}
		\sum_{\vec{\nod} \in \Nodes(F)}
		\car{\skelhastype{T}(\vec{\nod})}
		\lambda_T
		\\
		 & =
		\sum_{T \in \Types_{d,k}}
		\lambda_T
		\card{(\skelhastype{T} (\vec{\nod}))}
	\end{align*}
	
    The coefficient $\lambda_T$
    does not depend on the specfic $\vec{\nod}$
    such that $\skeltype(\vec{\nod}) = T$
    thanks to
    \cref{app:claim:bijection-types}
    and the fact that $\psi \in \INV$.
\end{proof}

The behavior of the formulas $\skelhastype{T}$ is much more
regular and enables us to extract pumping families that clearly
distinguishes different types. 
Namely, we are going to prove that given $k \geq 1$, $d \geq 0$,
a finite monoid $M$, and a morphism $\mu \colon A^* \to M$,
$\setof{\card{\skelhastype{T}}}{T \in \Types_{d,k}}$
is a $\Rel$-linearly independent family of functions
from $\forests{\mu}{d}$ to $\Rel$.

\begin{lemma}[Pumping Lemma]
	\label{lem:skel:iteration}
	For all $T \in \Types_{d,k}$,
	there exists a pumping family $(w^{\vec{X}}, F^{\vec{X}})$
	such that for every type $T' \in \Types_{d,k}$,
	$\card( \skelhastype{T'}) (F^{\vec{X}})$ is ultimately a
	$\Rel$-polynomial in $\vec{X}$
	that has non-zero coefficient for $X_1 \cdots X_n$
	if and only if $T = T'$.
\end{lemma}
\begin{proof}
	Let $T \in  \Types_{d,k}$ be a type, it is obtained
	as the type of some tuple $\vec{x}$ of independent leaves
	in some $F \in \forests{\mu}{d}$ factorizing a word $w$.
    Let $\nod_i \defined \partitionskel(x_i)$
    and $S_i \defined \Ske(\nod_i)$ for $1 \leq i \leq k$.
    Recall that 
    $\mu(\repword(S_i)) = \mu(\repword(\nod_i))$
    thanks to \cref{claim:skel-semantic-invariant}.
    As a consequence, $S_i$ is a subforest of $\nod_i$ that
    provides a valid $\mu$-forest of a subword of $\repword(\nod_i)$.

    Now, as $\nod_i$ cannot be the root of the forest $F$
    and is the highest ancestor of $x_i$ that is not a leftmost or rightmost
    child, it must be the immediate inner child of an idempotent node in $F$.
    As a consequence, $\mu(\repword(S_i)) = \mu(\repword(\nod_i))$ is an
    idempotent. Therefore, for ever $X_i \in \Nat$,
    the tree obtained by
    replacing $\nod_i$ with $X_i$ copies of $S_i$
    in $F$ is a valid $\mu$-forest.
    We write $F^{\vec{X}}$ for the forest $F$ where
    $\nod_i$ is replaced by $X_i$ copies of $S_i$. This is possible
    because
    the tuple $\vec{x}$ is composed of independent leaves, hence
    $\nod_i$ and $\nod_j$ are disjoint subtrees of $F$
    whenever $1 \leq i \neq j \leq k$.

    Hence,
    $F^{\vec{X}}$ is the factorization forest of the word
    $w^{\vec{X}} \defined 
        \alpha_0 (w_1)^{X_1} \alpha_1 \dots \alpha_{k-1} (w_k)^{X_k} \alpha_k$
    where $w_i = \repword(S_i)$, $\alpha_i = w[\skelend(\nod_i){+}1{:}
    \skelstart(\nod_i){-}1]$ for $2 \leq i \leq k-1$,
    $\alpha_0 = w[1{:}\skelstart(\nod_1) {-} 1]$,
    and $\alpha_k = w[\skelend(\nod_k) {+} 1{:} {|w|}]$
    are non-empty factors of $w$.

    We now have to understand the behavior of
    $\skelhastype{T'}$ over $F^{\vec{X}}$,
    for every $T' \in \Types_{d,k}$.
    To that end, let us consider $T' \in \Types_{d,k}$.
    Let us write $E$ for the set of nodes in $F^{\vec{X}}$ that
    are not appearing in any of the $X_i$ repetitions of $S_i$,
    for $1 \leq i \leq k$. The set $E$ has a size bounded independently
    of $X_1, \dots, X_k$.
    To a tuple $\vec{\nodb}$ such that
    $F^{\vec{X}} \models \skelhastype{T'}(\nodb)$,
    one can associate the mapping
    $\rho_{\vec{\nodb}} \colon \{1,\dots,k\} \to \{1, \dots, k\} \uplus E$,
    so that $\rho_{\vec{\nodb}}(i) = \nodb_i$ when $\nodb_i \in E$,
    and $\rho_{\vec{\nodb}}(i) = j$ when $\nodb_i$ is a node appearing in
    one of the $X_j$ repetitions of the skeleton $S_j$ 
    (there can be at most one $j$ satisfying this property).

    \begin{remark}
        \label{remark:appendix-skeltype-root}
        If $\skeltype(\vec{\nodb}) = T'$,
        and $\rho_{\vec{\nodb}}(i) = j$, then $\nodb_i$
        must be the root of one of the $X_j$ copies of $S_j$ in $F^{\vec{X}}$.
        Indeed, $\vec{\nod}$ is obtained as $\partitionskel(\vec{y})$ for
        some independent tuple $\vec{y}$ of leaves. Hence,
        $\nodb_i = \partitionskel(y_i)$ which belong to some copy of $S_j$, hence
        $\nodb_i$ must be the root of this copy of $S_j$, because $S_j$
        is a binary tree.
    \end{remark}

    Given a map $\rho \colon \{1, \dots k\} \to \{1, \dots, k\} \uplus E$ 
    and a tuple $\vec{X} \in \Nat^k$,
    we let
    $C_\rho(\vec{X})$
    be
    the set of tuples $\vec{\nodb}$ of nodes of $F^{\vec{X}}$
    such that $\skeltype(\vec{\nodb}) = T'$,
    and 
    such that $\rho_{\vec{\nodb}} = \rho$.
    This allows us to rewrite the number of such vectors as a finite
    sum:
    \begin{equation*}
        \card( \skelhastype{T'}(\vec{\nod})) (F^{\vec{X}})
        = \sum_{\rho \colon \{1, \dots, k\} \to \{1, \dots, k\} \uplus E}
        \card C_\rho(\vec{X})
    \end{equation*}

    \begin{claim}
        For every $\rho \colon \{1, \dots, k\} \to \{1, \dots, k\} \uplus E$,
        $\card{C_\rho(\vec{X})}$
        is ultimately a $\Rel$-polynomial
        in $\vec{X}$. Moreover,
        its
        coefficient for $X_1 \cdots X_k$
        is non-zero
        if and only if $\rho(i) = i$ for $1 \leq i \leq k$.
    \end{claim}
    \begin{claimproof}
        Assume that $C_\rho(\vec{X})$ is non-empty. Then
        choosing a vector $\vec{\nodb} \in C_\rho(\vec{X})$ is done by
        fixing the image of $\nodb_i$ to $\rho(i)$ when $\rho(i) \in E$,
        and selecting $p_j \defined |\rho^{-1}(\{j\})|$ non
        consecutive copies of $S_j$ among
        among the $X_j$ copies available.
        All nodes are accounted for since 
        \cref{remark:appendix-skeltype-root}
        implies that whenever $\nodb_i$ is in a copy of $S_j$,
        then $\nodb_i$ is the root of this copy, and since
        $\vec{\nodb}$ is independent, they cannot be direct siblings.

        The number of ways one can select $p$ non consecutive nodes
        in among $X$ nodes is
        (for large enough $X$) the binomial number $\binom{X-p+1}{p}$,
        as it is the same as selecting $p$ positions among $X-p+1$
        and then adding $p-1$ separators.

        As a consequence, the size of $C_\rho(\vec{X})$ is
        ultimately a product of $\binom{X_j - p_j + 1}{p_j}$ for
        the non-zero $p_j$,
        which is a $\Rel$-polynomial in $X_1, \dots, X_k$.
        Moreover, it has a non-zero coefficient for $X_1 \dots X_k$
        if and only if $p_j \neq 0$ for $1 \leq j \leq k$,
        which is precisely when $\rho(i) = i$.
    \end{claimproof}

    We have proven that $\card( \skelhastype{T'}) (F^{\vec{X}})$
    is a $\Rel$-polynomial in $X_1, \dots, X_k$, and 
    that the only term possibly having a non-zero coefficient for
    $X_1 \cdots X_k$ is $\card{C_{\id}(\vec{X})}$.
    Notice that if $\card{C_{\id}(\vec{X})}$
    is non-zero,
    we immediately conclude that
    $T = T'$.
\end{proof}

\begin{claim}
	\label{claim:polyzero}
	Let $P \in \mathbb{R}[X_1, \dots, X_n]$
	which evaluates to $0$ over  $\mathbb{N}^n$,
	then $P = 0$.
\end{claim}
\begin{claimproof}
	The proof is done by induction on the number $n$ of variables.
	If $P$ has one variable and $P_{|\Nat} = 0$,
	then $P$ has infinitely many roots and $P = 0$.
	Now, let  $P$ having $n+1$ variables, and such that
	$P(x_1, \dots, x_n, x_{n+1}) = 0$ for all $(x_1, \dots, x_{n+1}) \in
		\mathbb{N}^{n+1}$.
	By induction hypothesis, $P(X_1, \dots, X_n, x_{n+1}) = 0$
	for all $x_{n+1} \in \mathbb{N}$. 
    Hence for all $x_1, \dots, x_n \in \mathbb{R}$,
    $P(x_1, \dots, x_n, X_{n+1})$ is a polynomial with one free variable
    having infinitely many roots, hence $P(x_1, \dots, x_n, x_{n+1}) = 0$
    for every $x_{n+1} \in \mathbb{R}$.
    We have proven that
	$P = 0$.
\end{claimproof}

We now have all the ingredients to prove \cref{lem:skel:indep-iter-zero},
allowing us to pump functions built by counting independent tuples
of invariant formulas.

Let $k \geq 1$, and
$\findep$ be a linear combination of $\card{\psi_i \wedge \neg \msodep}$,
where $\psi_i \in \INV_k$.
Assume moreover that $\findep \neq 0$.
Thanks to \cref{lem:skel:decomp-indep-type},
every $\card{\psi_i \wedge \neg \msodep}$ can be
written as a linear combination of $\card{\skelhastype{T}(\vec{t})}$,
hence
$\findep = \sum_{T \in \Types_{d,k}} \lambda_T
\card{\skelhastype{T}}$, and the coefficients
$\lambda_T$ (now in $\Rel$) are computable.

Since $\findep \neq 0$, there exists $T \in \Types_{d,k}$
such that $\lambda_T \neq 0$. Using
\cref{lem:skel:iteration}, there exists a pumping family
$(w^{\vec{X}}, F^{\vec{X}})$ adapted to $T$.
In particular,
$f(F^{\vec{X}})$ is ultimately a $\Rel$-polynomial
in $\vec{X}$, and its coefficient
in $X_1 \cdots X_k$ is the sum of the coefficients
in $X_1 \cdots X_k$ of the polynomials
$\card{\skelhastype{T'}(F^{\vec{X}})}$ multiplied by $\lambda_{T'}$.
This coefficient is
non-zero if and only if $T=T'$. Hence, 
$f(F^{\vec{X}})$ is ultimately a $\Rel$-polynomial
with a non-zero coefficient for $X_1 \cdots X_k$.

As a side result, we have proven that
a linear combination of $\card{\skelhastype{T}}$ is the constant
function $0$ if and only if all the coefficient are $0$, which 
is decidable since one can enumerate all the elements of
$\Types_{d,k}$. For the converse implication,
one leverages \cref{claim:polyzero}: if one coefficient
is non-zero, then the polynomial $f(F^{\vec{X}})$
must be non-zero.

\subsection{Proof of \cref{lem:skel:lemmapoly}}

	Let $P,Q \in \mathbb{R}[X_1, \dots, X_n]$ be such
	that $|P| = \bigO(|Q|)$. We show that $\deg (P) \leq \deg (Q)$.

	If $P = 0$, then $\deg (P) \leq \deg (Q)$. Otherwise,
	let us write $P = P_1 + P_2$ with $P_1$ containing all the
	terms of degree exactly $\deg(P)$ in $P$.
    Because $|P| = \bigO(|Q|)$, there 
    exists $N \geq 0$ and $C \geq 0$
    such that $|P(x_1, \dots, x_n)| \leq C|Q(x_1, \dots, x_n)|$
    for all $x_1, \dots, x_n \in \Nat$
    such that $x_1, \dots, x_n \geq N$.

	Because $P_1$ is a non-zero polynomial,
    there exists a tuple $(x_1, \dots, x_n) \in \mathbb{N} \setminus \set{0}$
	such that $\alpha \defined P_1(x_1, \dots, x_n) \neq 0$
	(\cref{claim:polyzero}).
	Let us now consider
	$R(Y) \defined P(Y x_1, \dots, Y x_n) \in \Rea[Y]$,
	and $S(Y) \defined Q(Y x_1, \dots, Y x_n) \in \Rea[Y]$.
	Notice that $R(Y)$ has
	degree exactly $\deg(P)$ and its term of degree $\deg(P)$ 
	is $\alpha Y^{\deg(P)}$. Furthermore,
	$S(Y)$ is a polynomial in $Y$ of degree at most $\deg(Q)$,
    with dominant coefficient $\beta \neq 0$.
    We know that for $Y$ large enough,
    $|R(Y)| \leq C|S(Y)|$.
	Since $|R(Y)| \polysim{+\infty} |\alpha| Y^{\deg(P)}$,
	and $|S(Y)| \polysim{+\infty} |\beta| Y^{\deg(S)} \leq |\beta|
		Y^{\deg(Q)}$, we conclude that 
	$\deg(P) \leq \deg(Q)$.

\section{Proofs of \cref{sec:residual}}

\subsection{Proof of \cref{claim:push-pk}}

Let $k \ge 0$, $f \in \Ve_k$
and $u \in A^*$. We want to show
that $\push{f}{u} \in \Ve_k$.
Notice that for every $u$,
the map $u\square \colon w \mapsto uw$ is regular,
hence
$\push{f}{u} = f \circ (u\square)$
belongs to $\Ve_k$ thanks to \cref{prop:precompose}.

\subsection{Proof of \cref{claim:properties-resi}}

	The fact that  $\polysim{k}$ is an equivalence relation
	is obvious from the properties of $\Ve$. Furthermore
	 if $f \polysim{k} g$, then $f{-}g \in \Ve_k$,
	thus $\push{(f{-}g)}{u} = \push{f}{u} - \push{g}{u}  \in \Ve_k$
	by \cref{claim:push-pk}. Furthermore it is obvious that
	$\delta \cdot f \polysim{k} \delta \cdot g$,
	and if $f' \polysim{k} g'$ then
        $f + f' \polysim{k} g + g'$.
        
        It remains to show that
	$\push{(\car{L} \cauchy f)}{u}
	\polysim{k} (\push{\car{L}}{u}) \cauchy f$ for $L \subseteq A^*$
	and for this we proceed by induction on $|u|$.
	By expanding the definitions we note that
	$\push{(\car{L} \cauchy g)}{a}
	= (\push{{\car{L}}}{a}) \cauchy g + \car{L}(\movi) \times (\push{g}{a})$
for all $a \in A$.
By \cref{claim:push-pk} we get $\push{g}{a} \in \Ve_k$,
hence $\push{(\car{L} \cauchy g)}{a} \polysim{k} (\push{{\car{L}}}{a}) \cauchy g$.
The result follows since $\push{{\car{L}}}{a} = \car{a^{-1}L}$
and by \cref{theo:rational-polygrowth}.

\subsection{Proof of \cref{lem:resifini}}

We first note that $\push{(\delta f+\eta g)}{u}
	= \delta (\push{f}{u}) + \eta (\push{g}{u})$, for all $f,g : A^* \fonc \Rel$,
$\delta, \eta \in \Rel$ and $u \in A^*$.
Hence it suffices to show that \cref{lem:resifini}
holds on a set $S$ of functions such that $\Vect{\Rel}(S) = \Ve_k$.
For $k=0$, we can chose $S \defined \{\car{L} : L \text{~regular}\}$.
As observed above, we have $\push{\car{L}}{u} = \car{u^{-1}L}$
and the result holds since regular languages have
finitely many residual languages. For $k \ge 1$, we can choose
$S \defined \{\car{L} \cauchy g : g \in \Ve_{k-1}, L \text{ regular}\}$
by \cref{lem:vks:inductcauchy}. Let $\car{L} \cauchy g \in S$.
Then by \cref{claim:properties-resi} we get
$\push{(\car{L} \cauchy g)}{u} \polysim{k-1}
	(\push{\car{L}}{u}) \cauchy g = \car{u^{-1}L} \cauchy g $.
Since a regular language has finitely many residual languages,
there are finitely many $\polysim{k-1}$-equivalence classes for the
(function) residuals of $\car{L} \cauchy g$.

\subsection{Proof of \cref{lem:resitrans}}

	Let $f : A^* \fonc \Rel$ be a function
	such that $\Res(f) / \polysim{k-1}$. We apply \cref{algo:residual}, which computes
	the set of residuals of $f$ and the relations between them.
	The states of our machine are not labelled by the equivalence
	classes of $\Res(f)/\polysim{k-1}$, but directly by some elements of
	$\Res(f)$.
	Remark that the labels on the transitions are of the form $\push{f}{w} - \push{f}{v}$
	when $\push{f}{w} \polysim{k-1} \push{f}{v}$, hence are in
	$\Vect{\Rel}(\Res(f)) \cap \Ve_{k-1}$ by definition of $\polysim{k-1}$
	(observe that the construction of these labels is effective and that equivalence
	of residuals is decidable if we start from $f \in \Ve_{k}$).
	Now, let us justify the correctness and
	termination of \cref{algo:residual}.

	First, we note that it maintains two sets $O$ and $Q$
	such that $O \uplus Q \subseteq \Res(f)$ and
	for all $f,g \in O \uplus Q$ we have $f \neq g \implies f \not \polysim{k-1} g$.
	Hence the algorithm terminates since $\Res(f)/\polysim{k-1}$ is finite and
	$Q$ increases at every loop.
	At the end of its execution, we have for all $q \in Q$ and $a \in A$,
	that $\delta(q,a) \polysim{k-1} \push{q}{a}$ and
	$\lambda(q,a) = \push{q}{a} - \delta(q,a)$.

	Let us show by induction on $n \ge 0$ that for all $a_1 \cdots a_n \in A^*$,
	if $q_0 \fonc^{a_1} q_1 \fonc^{a_2} \cdots \fonc^{a_n} q_n$ is the run labelled by $a_1 \cdots a_n$
	in the underlying automaton , and $g_1 \cdots g_n$
	are the functions which label the transitions, we have $q_n \polysim{k-1} \push{f}{a_1 \cdots a_n}$
	and for all $w \in A^*$, $f(a_1 \cdots a_nw) = \sum_{i=2}^n g_i(a_i \cdots a_n w) + q_n(w)$.
	For $n = 0$ the result is obvious because $q_0 = f$.
	Now, assume that the result holds for some $n \ge 0$ and
	let $a_1 \cdots a_{n} a_{n+1} \in A^*$. Let
	$q_0 \fonc^{a_1} q_1 \fonc^{a_2} \cdots \fonc^{a_{n+1}} q_{n+1}$ be the run
	and $g_1 \cdots g_{n+1}$ be the labels of the transitions.
	Since  $q_n \polysim{k-1} \push{f}{a_1 \cdots a_{n}}$ (by induction) we get
	$\push{q_n}{a_{n+1}} \polysim{k-1} \push{f}{a_1 \cdots a_{n}a_{n+1}}$ by
    \cref{claim:properties-resi}.
	Because $q_{n+1} = \delta(q_n,a_{n+1}) \polysim{k-1} \push{q_n}{a_{n+1}}$,
	then $q_{n+1}  \polysim{k-1} \push{f}{a_1 \cdots a_{n}a_{n+1}}$.
	Now, let us fix $w \in A^*$. We have
	$f(a_1 \cdots a_{n}a_{n+1}w) = \sum_{i=2}^n g_i(a_i \cdots a_n a_{n+1}u) + q_n(a_{n+1}w)$
	by induction hypothesis.
	But since $g_{n+1} = \lambda(q_n, a_{n+1})
	= \push{q_n}{a_{n+1}} - \delta(q_n, a_{n+1})
		= \push{q_n}{a_{n+1}} -q_{n+1} $
		we get $ q_n(a_{n+1} w) = g_{n+1}(w) + q_{n+1}(w)$.
	We conclude the proof that \cref{algo:residual} provides
	a $k$-residual transducer for $f$ by considering $w = \movi$
	and the definition of $F$.

\subsection{Proof of \cref{cor:Vk-trans}}

\cref{lem:resitrans} shows that any function from
$\Ve_k$ is computed by its $k$-residual transducer
(which is in particular a $\Ve_{k-1}$-transducer). Conversely,
given a $\Ve_{k-1}$-transducer computing $f$, it is easy to write $f$
as a linear combination of elements of the form $\car{L} \cauchy g$
(see e.g. \cref{pro:decompo}),
where $g$ is the label of a transition, thus $f \in \Ve_{k-1}$.

\subsection{Proof of \cref{cor:faux}}

Every map in $\Ve_k$ has finitely many residuals up to $\polysim{k-1}$
thanks to \cref{lem:resifini}. We now prove the converse implication.
Let $f$ such that $\Res(f) / \polysim{k-1}$ is finite. By \cref{lem:resitrans}
there exists a $k$-residual transducer of $f$
(which is in particular a $\Ve_{k-1}$-transducer).
Thanks to \cref{cor:Vk-trans}, it follows that $f \in \Ve_k$.

\section{Proofs of \cref{sec:aperiodic}}

\subsection{Proof of \cref{ex:1poly-aperiodic}}

Let $L$ be a regular language such that $\car{L}$ is ultimately $1$-polynomial.
Then, for every $u,w,v \in A^*$,
there exists a polynomial $P \in \Rat[X]$, such that
$\car{L}(uw^X v) = P(X)$ for $X$ large enough. This implies
that $P$ is a constant polynomial, and in particular
$\car{L}(uw^{X+1}v) = \car{L}(uw^{X}v)$ for $X$ large enough.
As a consequence, the syntactic monoid of $L$ is aperiodic,
thus $L$ is star-free \cite{schutzenberger1965finite}.
Conversely, assume that $L$ is star-free.
It is recognized by a morphism $\mu$
into an aperiodic finite monoid $M$.
Because $M$ is aperiodic, for every
$x \in M$, $x^{|M|+1} = x^{|M|}$. 
Hence, for all  $\alpha_0, w_1, \alpha_1, \dots, w_{\ell}, \alpha_{\ell} \in A^*$,
$\car{L}(\alpha_0 w_1^{X_1} \alpha_1 \cdots w_{\ell}^{X_{\ell}}\alpha_{\ell})$
is constant for $X_1, \dots, X_\ell \geq |M|$
since it only depends on the image
$\mu(\alpha_0 w_1^{X_1} \alpha_1 \cdots w_{\ell}^{X_{\ell}}\alpha_{\ell})$.

\subsection{Proof of \cref{lem:aper:counterfreefo}}
\label{pro:decompo}

Let $\trans = (A,Q, q_0, \delta, \lambda)$ be a
counter-free $\We_{k-1}$-transducer computing a function
$f : A^* \fonc \Rel$. Since the deterministic automaton
$(A,Q, q_0, \delta)$ is counter-free,
then by \cite{mcnaughton1971counter} for all $q \in Q$ the language
$L_q \defined \{u : \delta(q_0,u) = q\}$ is star-free. So is $L_q a$ for all $a \in A$.
Now observe that:
\begin{equation*}
	\begin{aligned}
		f = \sum_{\substack{q \in Q \\ a \in A}} \car{L_q a} \cauchy \lambda(q,a).
	\end{aligned}
\end{equation*}
We conclude thanks to \cref{eq:cauchys-starfree}.

\subsection{Proof of \cref{lem:up-cf}}

	Let $k \ge 0$. Let $f \in \Ve_k$ which is
    ultimately $1$-polynomial
	and $\trans = (A, Q, q_0, \delta, \Oras, \lambda, F)$ be a $k$-residual transducer of $f$.
	Since ultimate $1$-polynomiality
    is preserved under taking
	linear combinations and residuals, the function labels
	of $\trans$ are
    ultimately $1$-polynomial
    (by definition of a $k$-residual
	transducer). It remains to show that $\trans$ is counter-free.

	Let $\alpha, w \in A^*$ and suppose that $\delta(q_0, \alpha)
		= \delta(q_0, \alpha w^n)$ for some $n \ge 1$.
	We want to show that $\delta(q_0, \alpha w) = \delta(q_0, \alpha)$.
	Since $\delta(q_0, \alpha) = \delta(q_0, \alpha w^{n X})$
	and $\delta(q_0, \alpha w) = \delta(q_0, \alpha \alpha w^{n X+1})$
	for all $X \ge 1$, it is sufficient to
	show that we have $\delta(q_0, \alpha w^{n X + 1}) = \delta(q_0, \alpha w^{n X})$
	for some $X \ge 1$.

	Let $M \ge 1$ given by \cref{def:ultimately-polynomial-general} for
	the 
    ultimate $1$-polynomiality
    of $f$.
	We want to show that $(\push{f}{\alpha w^{n M +1}}) \polysim{k-1}
		(\push{f}{\alpha w^{n M}})$, i.e. $| (\push{f}{\alpha w^{n M +1}})(w) -
		(\push{f}{\alpha w^{n M}})(w)| = \bigO(|w|^{k-1})$ since the residuals
	belong to $\Ve$. For this, let us pick any $\alpha_0, w_1, \alpha_1, \cdots, w_k, \alpha_{k} \in A^*$.
	By \cref{thm:skel:pebblemin}, it is sufficient to show that:
	\begin{equation*}
        \begin{array}{rl}
            &|(\push{f}{\alpha w^{n M}} - \push{f}{\alpha w^{nM +1}})
			(\alpha_0 w_1^{X_1} \cdots w_{k}^{X_{k}} \alpha_{k}))| \\
            &= \bigO((X_1 + \cdots + X_k)^{k-1})
		\end{array}
	\end{equation*}
	Because $f$ is
    ultimately $1$-polynomial, for all $X, X_1, \cdots, X_{k} \ge M$,
	$f(\alpha w^{X} \alpha_0 w_1^{X_1} \cdots w_{k}^{X_{k}} \alpha_{k})$
	is a polynomial $P(X, X_1, \dots, X_{k})$. Our goal is
	to show that $|P(nM, X_1, \dots, X_{k})
		- P(nM+1, X_1, \dots, X_{k})| = \bigO( |X_1 + \cdots + X_k|^{k-1})$.
	Since $f \in \Ve_k$, we have $|P(X, X_1, \dots, X_{k})| = \bigO(|X + X_1 +
		\cdots + X_k|^k)$.
	Thus by \cref{lem:skel:lemmapoly}, $P$ has degree at most $k$,
	hence it can be rewritten under the form
	$P_0 + X P_1 + \cdots + X^{k} P_k$ where $P_i(X_1, \dots, X_k)$
	has degree at most $k-i$. Therefore:
	\begin{equation*}
        \begin{array}{rl}
             &|P(nM, X_1, \dots, X_{k})
			- P(nM+1, X_1, \dots, X_{k})|
            \\
			 & = \left| \sum_{i=1}^k P_i( X_1, \dots, X_{k}) ((nM)^i - (nM+1)^i) \right|
             \\
             & \le \sum_{i=1}^k |P_i( X_1, \dots, X_{k})| (nM+1)^i
		\end{array}
	\end{equation*}
	since the term $P_0$ vanishes when doing the subtraction.
	The result follows since the polynomials $P_i$
	for $1 \le i \le k$ have degree at most $k{-}1$.

\subsection{Proof of \cref{prop:mikolaj2}}

The proof of the proposition is essentially the same as \cref{prop:mikolaj1}
by noticing that everything remains $\FO$-definable. We will \underline{underline}
the parts where the two proofs differ, and in particular
when using stability properties of star-free polyregular functions.

\newcommand{\ul}[1]{\underline{#1}}

  We first show that any \ul{star free} $\Rel$-polyregular
    function can be written under the form $\polysum \circ g$
    where $g : A^* \fonc \{\pm 1\}^*$ is \ul{star-free} polyregular.
    This is a consequence of the following claims.

	\begin{claim}
        For all $\varphi \in \ul{\FO}$, there exists a \ul{star-free} polyregular function 
        $f \colon A^* \to \set{\pm 1}^*$ such that  $\card \varphi = \polysum \circ
        f$.
	\end{claim}
    \begin{proof}
        \ul{Star-free} polyregular functions are characterized
        in  \cite[Theorem 7]{bojanczyk2019string}
        as the functions computed by 
        (multidimensional) \ul{$\FO$}-interpretations.
        Recall that an \ul{$\FO$}-interpretation
        of dimension $k \in \Nat$
        is given by a $\ul{\FO}$ formula $\varphi_{\leq}(\vec{x},\vec{y})$ 
        defining a total ordering
        over $k$-tuples of positions,
        a $\ul{\FO}$ formula $\varphi^{\Dom}(\vec{x})$ that selects
        valid positions, and $\ul{\FO}$ formulas $\varphi^{a}(\vec{x})$ that
        place the letters over the output word
        \cite[Definition 1 and 2]{bojanczyk2019string}.
        In our specific situation, letting
        $\varphi_{\leq}$ be the usual lexicographic ordering of positions 
        (which is $\ul{\FO}$-definable)
        and placing the letter $1$ over every element of the output
        is enough: the only thing left to do is select
        enough positions of the output word.
        For that, we let $\varphi^{\Dom}$ be defined as $\varphi$ itself.
        It is an easy check that this $\ul{\FO}$-interpretation
        precisely computes $1^{f(w)}$ over $w$,
        hence computes $f$ when post-composed with $\polysum$.
    \end{proof}
	
	\begin{claim}
        The set $\setof{\polysum \circ f}{ 
        f \colon A^* \to \set{\pm 1}^* \text{ \ul{star-free} polyregular}}$
        is closed under sums and external $\Rel$-products.
	\end{claim}
    \begin{proof}
        Notice that $\polysum \circ f + \polysum \circ g = \polysum \circ (f
        \cdot g)$ where $f \cdot g(w) \defined f(w) \cdot g(w)$.
        As \ul{star-free} polyregular functions are closed under concatenation
        \cite{bojanczyk2018polyregular},
        the set of interest is closed under sums.
        To prove that it is closed under external $\Rel$-products,
        it suffices to show that it is closed under negation.
        This follows because one can permute the
        $1$ and $-1$ in the output of a \ul{star-free} polyregular function
        (\ul{star-free} polyregular functions are closed under post-composition
        by a morphism \cite[Theorem 2.6]{bojanczyk2018polyregular}).
    \end{proof}

    Let us consider a \ul{star-free} polyregular function
	$g \colon A^* \to \set{ \pm 1 }^*$.
	The maps $g_+ \colon w \mapsto |g(w)|_1$
	and $g_- \colon w \mapsto |g(w)|_{-1}$
    are \ul{star-free} polyregular functions with unary output
	(since they correspond to a post-composition by
	the \ul{star-free} polyregular function which removes some letter, and polyregular functions are
	closed under post-composition by a regular function \cite{bojanczyk2018polyregular}).
	Hence $g_-$ and $g_+$ are star-free polyregular functions with unary output,
    a.k.a. \underline{star-free} $\Nat$-polyregular functions.
    As a consequence,
	${\polysum} \circ g = g_+ - g_-$ lies in $\We$.

\subsection{Proof of \cref{prop:counting-mso-variables}}

	\Cref{it:FOX} $\Rightarrow$ \cref{it:MSOX} is obvious.
	For \cref{it:MSOX} $\Rightarrow$ \cref{it:ratrat}, it is sufficient
	to show that if $\varphi(X_1, \dots, X_n)$ is an $\MSO^X$ formula,
	then $\card \varphi$ is a $\Rel$-polyregular function. We show the
	result for $n=1$, i.e. for a formula $\varphi(X)$. Let us
	define the language $L \subseteq (A \times \{0,1\})^*$ such that
	$(w,v) \in L$ if and only if $w \models \phi(S)$
	where $S \defined \setof{1 \le i \le |w|}{v[i] = 1}$.
	Using the classical correspondence between $\MSO$
	logic and automata (see e.g. \cite{thomas1997languages}),
	the language $L$ is regular, hence it is computed by a finite deterministic
	automaton $\Aper$. Given a fixed $w \in A^*$, there exists a bijection
	between the accepting runs of $\Aper$ whose first component is $w$
	and the sets $S$ such that $w \models \phi(S)$. Consider the (nondeterministic)
	$\Rel$-weighted automaton $\Aper'$ (this notion is equivalent to $\Rel$-linear
	representations, see e.g. \cite{berstel2011noncommutative})
	obtained from $\Aper$ by removing the second component of the input,
	adding an output $1$ to all the transitions of $\Aper$, and giving the initial values $1$
	(resp. final values $1$) to the initial state (resp. final states) of $\Aper$.
	All other transitions and states are given the value $0$.
	Given a fixed $w \in A^*$, it is easy to see that $\Aper'$
	has exactly $\card \varphi(w)$ runs labelled by $w$ whose product
	of the output values is $1$ (and the others have product $0$).
	Thus $\Aper$ computes $\card \varphi$.
    This proof scheme adapts naturally to the case where $n \geq 1$.
	
	For \cref{it:ratrat} $\Rightarrow$ \cref{it:FOX}, let us consider
    a linear representation $(I, \mu, F)$ of a $\Rel$-rational series.
    \begin{claim}
        Without loss of generality, one can assume that $\mu(A^*) \subseteq
        \Mat{n,n}(\set{0,1})$, at the cost of increasing the dimension of the matrices.
    \end{claim}
    \begin{proof}[Proof Sketch]
        Let
        $N \defined \min(1, \max \setof{ |\mu(a)_{i,j}|}{ a \in A, 1 \leq i,j \leq n
        })$,
        we define the new dimension of our system to be
        $m \defined n \times N \times 2$. As a notation,
        we assume that matrices in $\Mat{m,m}$ have their 
        rows and columns indexed by $\set{1, \dots, n} \times \set{1, \dots, N}
        \times \set{\pm}$.
        For all $a \in A$, let us define  $\nu(a) \in \Mat{m,m}$
        as follows: for all $1 \leq i,j \leq n$,
        $1 \leq v, v' \leq N$
        \begin{align*}
            \nu(a)_{(i,v,+), (j,v',+)} &= 
            \begin{cases}
                1 & \text{ if } |\mu(a)_{i,j}| \geq v' \wedge 0 < \mu(a)_{i,j} \\
                0 & \text{ otherwise }
            \end{cases}
            \\
            \nu(a)_{(i,v,+), (j,v',-)} &= 
            \begin{cases}
                1 & \text{ if } |\mu(a)_{i,j}| \geq v' \wedge 0 > \mu(a)_{i,j} \\
                0 & \text{ otherwise }
            \end{cases}
            \\
            \nu(a)_{(i,v,-), (j,v',-)} &= 
            \begin{cases}
                1 & \text{ if } |\mu(a)_{i,j}| \geq v' \wedge 0 < \mu(a)_{i,j} \\
                0 & \text{ otherwise }
            \end{cases}
            \\
            \nu(a)_{(i,v,-), (j,v',+)} &= 
            \begin{cases}
                1 & \text{ if } |\mu(a)_{i,j}| \geq v' \wedge 0 > \mu(a)_{i,j} \\
                0 & \text{ otherwise }
            \end{cases}
        \end{align*}
        Let us now adapt the final vector by
        defining
        for every $1 \leq i \leq n$, $1 \leq v \leq N$,
        $F'_{(i,v,+)} \defined \max(0, F_i)$,
        and
        $F'_{(i,v,-)} \defined -\min(0, F_i)$.
        For the initial vector, let us
        define for every $1 \leq i \leq n$,
        $I'_{(i,1,+)} = I_i$ and $I'_{(i,1,-)} = -I_i$,
        and let $I'$ be zero otherwise.
        It is then an easy check
        that $(I', \nu, F')$ computes the same function as $(I, \mu, F)$.
    \end{proof}

    As a consequence,
    $I \mu(w) F = \sum_{i,j} I_i \mu(w)_{i,j} F_j$, let us now
    rewrite this sum as a counting $\MSO$ formula with set free variables.

    For all $1 \leq i,j \leq n$,
    one can write an $\MSO$ formula $\psi_{i,j} (x)$
    such that
    for all $1 \leq p \leq |w|$,
    $w \models \psi_{i,j} (p)$ if and only if $\mu(w[p])_{i,j} = 1$.
    Furthermore,
    for all $1 \leq i,j \leq n$,
    one can write an $\MSO$ formula $\theta_{i,j}$
    with variables $X_p^{\mathsf{in}}, X_p^{\mathsf{out}}$ for $1 \leq p \leq n$
    such that a word $w$ satisfies
    $\theta_{i,j}$ whenever for every position $x$ of $w$
    there exists a unique pair $1 \leq p,q \leq n$ such that
    $x \in X_p^{\mathsf{in}}$ and $x \in X_q^{\mathsf{out}}$,
    if $x \in X_p^{\mathsf{out}}$ then $(x+1) \in X_p^{\mathsf{in}}$,
    the first position of $w$ belongs
    to $X_i^{\mathsf{in}}$ and $X_i^{\mathsf{out}}$, and the last position of $w$
    belongs to $X_j^{\mathsf{in}}$ and $X_j^\mathsf{out}$.
    \begin{align*}
        \mu(w)_{i,j}
    &= \sum_{s \colon \set{1, \dots, k-1} \to \set{1, \dots, n} }
    \mu(w[1])_{i, s(1)}
    \mu(w[|w|])_{s(k-1),j}
    \prod_{k =
    2}^{|w|-1} \mu(w[k])_{s(k), s(k+1)} 
    \\
    &=
    \card{
        \underbrace{
        \left(
        \theta_{i,j}
        \wedge
        \forall x.
        \bigwedge_{1 \leq i,j \leq n}
        (x \in X_i^{\mathsf{in}} \wedge x \in X_j^{\mathsf{out}})
        \implies \psi_{i,j}(x)
        \right)
    }_{\defined \tau_{i,j}}
    }(w)
    \end{align*}
    We have proven that $I \mu(w) F$
    is a $\Rel$-linear combination of
    the counting formulas $\tau_{i,j}$
    via
    $I \mu(w) F = \sum_{i,j} I_i F_j \cdot \card{\tau_{i,j}} (w)$.
    Notice that all the formulas used never introduced set quantifiers,
    hence the formulas belong to $\FO$ and have $\MSO$ free variables.